\documentclass{article}

\usepackage{arxiv}

\usepackage{graphicx}%
\usepackage{multirow}%
\usepackage{amsmath,amssymb,amsfonts}%
\usepackage{amsthm}%
\usepackage{mathrsfs}%
\usepackage[title]{appendix}%
\usepackage{xcolor}%
\usepackage{textcomp}%
\usepackage{manyfoot}%
\usepackage{booktabs}%
\usepackage{algorithm}%
\usepackage{algorithmicx}%
\usepackage{algpseudocode}%
\usepackage{listings}%
\usepackage{subcaption}

\newtheorem{theorem}{Theorem}
\newtheorem{proposition}[theorem]{Proposition}%

\newtheorem{definition}{Definition}%

\newtheorem{corollary}[theorem]{Corollary}

\newcommand{\R}{\mathbb{R}}

\newcommand{\X}{\mathcal{X}}
\newcommand{\W}{\mathcal{W}}

\newcommand{\y}{\boldsymbol{y}}
\newcommand{\x}{\boldsymbol{x}}
\newcommand{\z}{\boldsymbol{z}}
\newcommand{\e}{\boldsymbol{e}}
\newcommand{\w}{\boldsymbol{w}}
\newcommand{\p}{\boldsymbol{p}}
\newcommand{\q}{\boldsymbol{q}}

\newcommand{\K}{\boldsymbol{K}}
\newcommand{\D}{\boldsymbol{D}}
\newcommand{\M}{\boldsymbol{M}}
\newcommand{\I}{\boldsymbol{I}}

\newcommand{\prox}{\mathrm{prox}}

\usepackage{enumitem} 

\title{Space-Variant Total Variation boosted by learning techniques in few-view tomographic imaging}

\author{Elena Morotti \\
Department of Political and Social Sciences, \\ University of Bologna, Italy. \\ \texttt{elena.morotti4@unibo.it}. \\
\And Davide Evangelista \\
Department of Computer Science and Engineering, \\ University of Bologna, Italy. \\
\And Andrea Sebastiani \\
Department of Physics, Informatics and Mathematics, \\ University of Modena and Reggio Emilia, Italy \\
\And Elena Loli Piccolomini \\
Department of Computer Science and Engineering, \\ University of Bologna, Italy.}

\date{}

\begin{document}
\maketitle

\begin{abstract}
This paper focuses on the development of a space-variant regularization model for solving an under-determined linear inverse problem. The case study is a medical image reconstruction from few-view tomographic noisy data. 
The primary objective of the proposed optimization model is to achieve a good balance between denoising and the preservation of fine details and edges, overcoming the performance of the popular and largely used Total Variation (TV) regularization through the application of appropriate pixel-dependent weights.
The proposed strategy leverages the role of gradient approximations for the computation of the space-variant TV weights.
For this reason, a convolutional neural network is designed, to approximate both the ground truth image and its gradient using an elastic loss function in its training. \\
Additionally, the paper provides a theoretical analysis of the proposed model, showing the uniqueness of its solution, and illustrates a Chambolle-Pock algorithm tailored to address the specific problem at hand.
This comprehensive framework integrates innovative regularization techniques with advanced neural network capabilities, demonstrating promising results in achieving high-quality reconstructions from low-sampled tomographic data.
\end{abstract}


\section{Introduction}
\label{sec:intro}

In this work, we consider linear inverse problems of the form:
	\begin{equation}\label{eq:inverse_problem}
		\y^\delta = \K \x^{GT} + \e, \quad || \e ||_2 \leq \delta,
	\end{equation}
where $\x^{GT} \in \X \subseteq \R^n$ is the (unknown) ground truth image, $\K \in \R^{m \times n}$ is an under-sampled linear operator, i.e.  $m \leq n$, $\y^\delta \in \R^m$ is the acquired, noisy datum, and $\e \in \R^m$ is the noise, whose norm is bounded by $\delta$.
Note that, since $m \leq n$, the kernel of $\K$ is generally non-trivial, which implies that \eqref{eq:inverse_problem} admits infinite solutions. Moreover, if $\K$ comes from the discretization of an ill-posed integral operator (as it typically happens in imaging), the solutions of \eqref{eq:inverse_problem} do not continuously depend on the data $\y^\delta$. To handle the scarcity of data and the ill-posedness of the problem, a good solution of \eqref{eq:inverse_problem} can be computed by minimizing a regularized problem of the form:
	\begin{equation}\label{eq:variational_formulation}
        \x^* = \arg\min_{\x \in \X} \mathcal{J}(\x, \y^\delta) := \mathcal{F}(\x, \y^\delta)  + \lambda \mathcal{R}(\x),
	\end{equation}
where $\mathcal{F}(\x, \y^\delta)$ is the data-fidelity term 
is utilized to ensure that $x^*$
  aligns with the acquisition model \eqref{eq:inverse_problem}
and $\mathcal{R}(\x)$ is the regularization term, enforcing stability and uniqueness of the solution by considering prior knowledge on $\x^*$. The scalar $\lambda>0$ controls the contribution of the regularization term over the fidelity one in the overall objective function $\mathcal{J}(\x, \y^\delta)$.

In traditional regularization approaches, $\mathcal{R}(\x)$ is applied uniformly to the whole image, i.e. its formulation equally weights the contributions across the pixels. We refer to this case as {\em global regularization}. 
Global regularization is simple to implement and computationally low demanding; 
on the other hand, it treats the entire image as a homogeneous entity making the setting the optimal weight $\lambda$ not trivial conceptually, because a value may not adapt well to varying characteristics within the image.
For instance, a considerably high value of $\lambda$ may effectively denoise the resulting image within flat regions, yet it could excessively smooth it, potentially leading to the loss or blurring of crucial details in certain areas. Conversely, a slightly lower value might adequately preserve details but fail to sufficiently eliminate noise in uniform areas.

To increase the flexibility of the regularization and better balance the trade-off between denoising and preserving fine details, a good strategy is the {\em space-variant regularization}. It allows for adapting regularization strength based on local image characteristics. This can enhance the preservation of fine details and edges without losing denoising efficacy, but its implementation is often more complex and computationally intensive than for global methods. In addition, its performance might be more sensitive to noise and artifacts, particularly if the local characteristics are not accurately estimated. In particular, the selection of optimal local regularization parameters is very challenging. 


In this work, we explore an imaging application focused on reconstructing medical images from few tomographic measurements. 
In this case, $\K$ comes from the discretization of the 2-dimensional fan-beam Radon transform and fits the aforementioned properties of ill-posedness. 
The data perturbations can be modeled with a white additive Gaussian noise, as in \eqref{eq:variational_formulation}, hence it is common to set $\mathcal{F}(\x,\y^\delta) = \frac{1}{2} || \K\x-\y^\delta||_2^2$ as the data-fitting term. 
Our task-oriented regularizing prior should preserve diagnostically significant features as well as remove noise. As medical images often contain relevant information in the form of low-contrast objects and boundaries between different anatomical tissues, a common choice for $\mathcal{R}$ is the Total Variation (TV) function. In fact, TV is known to be effective in preserving edges because it penalizes the variation or changes in pixel intensity by promoting sparsity in the gradient image domain \cite{sidky2014cttpv, piccolomini2021model, friot2022iterative}.
When globally applied, the drawback of TV regularization is the potential introduction of piecewise constant artifacts, which could impact the interpretation of medical information. It is particularly noticeable in regions with smooth intensity variations, as TV regularization tends to oversimplify these areas, resulting in blocky or stair-step artifacts. 
Furthermore, considering the very restricted number of projections acquired in few-view X-ray examinations (in contrast to conventional imaging protocols), the blocky counteraction can severely promote streaking artifacts caused by the lack of views in the acquisition process.

\paragraph{Contributions.} The contributions of this paper are twofold, both theoretical and algorithmic.

Firstly,  we demonstrate the uniqueness of the TV-based regularized solution, in the special case where $\K$ is under-determined. It is widely acknowledged that the solution is unique when $\K$ is over-determined; however, to the best of our knowledge, there is no proof in the under-determined case.\\
Secondly, we introduce a novel model for space-variant Total Variation (TV) regularization. In this model, the TV function's weights are determined by evaluating the gradient magnitudes of a pre-computed image.  We demonstrate that this image should closely approximate the ground truth, particularly within its gradient domain, to get an accurate final solution.  

On the algorithmic front, we present several methods for generating the pre-calculated image as a coarse reconstruction from the given observations.
When provided with only a single datum $\y^\delta$, we can employ either an analytical or a variational approach for a tomographic problem. 
Conversely, in cases where the medical scenario provides a dataset comprising consistent pairs of data and accurate reconstructions, we can harness the power of neural networks to generate images with desired features. We have developed a convolutional neural network trained to approximate both the ground truth image and its gradient, utilizing an elastic loss function.\\
Finally, we tailor the widely recognized Chambolle-Pock algorithm to address the resulting minimization problem.
Extensive numerical simulations  conducted on both a synthetic image and a real dataset assess the relevance of the model for this application of tomographic imaging.

The paper is organized as follows.
In Section \ref{sec:PbStatem}, we initially establish the uniqueness of the solution to a weighted Total Variation regularization model, followed by an exposition of our approach for determining the weights.  Subsequently, we outline the algorithmic aspects in Section \ref{sec:Algo} and
we conduct numerical experiments in Section \ref{sec:numexp} to evaluate the effectiveness of our proposal across various implementations. 
Finally, we discuss our approach and set the conclusions in Section \ref{sec:Concl}, whereas we refer to the Appendix \ref{sec:Appendix} for further theoretical aspects.

\section{Space-Variant regularization models} \label{sec:PbStatem}

In this Section, we state our regularized model and derive important theoretical results about its convergence to a unique optimal solution. 

To simplify the readability of the upcoming results, we introduce here some notations. 
When not specified, we always assume vectors $\x$ to lie in $\X \subseteq \R^n$, which we assume to be convex. Thus, for any $\x, \y \in \X$, the symbol $\langle \x, \y \rangle$ represents the standard inner product in $\R^n$. Moreover, if $\x, \y \in \R^n$, we say that $\x \leq \y$ if and only if $\x_i \leq \y_i$ for any $i = 1, \dots, n$. If $\boldsymbol{A} \in \R^{m \times n}$ is a linear operator, then $\boldsymbol{A}^T$ represents the transpose of $\boldsymbol{A}$. If $C \subseteq \R^n$, then $\iota_C(\x)$ is the characteristic function of $C$, defined as:
\begin{align}
    \iota_C(\x) = \begin{cases}
        0 \quad & \mbox{if } \x \in C, \\
        +\infty \quad & \mbox{if } \x \notin C.
    \end{cases}
\end{align}
Finally, we define the 2-dimensional discrete gradient operator $\D: \R^n \to \R^{2n}$, such that:
\begin{align}\label{eq:Dx_definition}
    \D\x = \begin{bmatrix}
        \D_h \x \\ \D_v \x
    \end{bmatrix} \in \R^{2n},
\end{align}
where $\D_h, \D_v \in \R^{n \times n}$ are the discrete central differences operators associated with the horizontal and vertical derivatives, respectively, and $| \D \x | \in \R^n$ is the gradient-magnitude image of $\x$, defined as:
	\begin{align}\label{eq:gradient_magnitude}
		\left( | \D \x | \right)_i =\sqrt{\left( \D_h \x \right)_i^2 + \left( \D_v \x \right)_i^2}.
	\end{align}
Note that, from \eqref{eq:Dx_definition}, it holds that $(\D\x)_i = (\D_h\x)_i$ for $i = 1, \dots, n$, while $(\D\x)_i = (\D_v\x)_{i-n}$ for $i=n+1, \dots, 2n$.

\subsection{Regularization based on Total Variation}\label{ssec:OurModel}

Space-variant (or adaptive) regularization functions, which operate distinctively within each pixel of the image, have been proposed in the literature, primarily grounded in the Total Variation approach. 
The popular isotropic Total Variation operator $TV(\x)$ is the Rudin, Osher, Fatemi prior \cite{rudin1992nonlinear} defined as:
	\begin{align}\label{eq:TV_definition}
		TV(\x) = || \D \x ||_{2, 1} := \sum_{i=1}^n \sqrt{\left( \D_h \x \right)_i^2 + \left( \D_v \x \right)_i^2}.
	\end{align}

In a recent review by \cite{pragliola2023and}, the authors demonstrate the efficacy of space-variant TV-based regularization in addressing the limitations of TV in characterizing local features. Indeed, a widely employed technique to derive space-variant models involves weighting the regularization function pixel-wise. They suggest various sparsifying regularization functions with weights derived from a Bayesian interpretation of the model.
 In other papers, these weights are typically determined based on either an estimation of the data noise \cite{hintermuller2017optimalI, hintermuller2017optimalII, bortolotti2016uniform, cascarano2023constrained}, by an assessment of the image scaling \cite{grasmair2009locally, chen2010adaptive} or by means of specifically trained neural networks \cite{cuomo2023speckle}.
An alternative approach entails utilizing diverse measures of the image gradient, as exemplified in the Total p-Variation (TpV) regularization method, which typically is:
    \begin{align}\label{eq:TpV}
	|| \D \x ||_{2, p}^p := \sum_{i=1}^n \left( \sqrt{\left( \D_h \x \right)_i^2 + \left( \D_v \x \right)_i^2} \right)^p
    \end{align}
with $0<p<1$.    
In this scenario, the parameter $p$ can be adapted for each pixel of the image \cite{chen2010adaptive, chen2021non}. \\
All the referenced papers employ regularization techniques to address denoising and deblurring problems, wherein the observed image (datum) and the reconstructed image (solution) belong to the same space. However, our application differs in that the data (sinogram) resides in a lower-dimensional space of the solution; thus, the noise present in the sinogram varies significantly from the noise inherent in the reconstructed image. Indeed, streaking artifacts also impact the reconstructions, originating from the sparsity of the views.
Given that several of the referenced algorithms adjust adaptive weights at each iteration of the solver, leading to heightened computational overhead, they are unsuitable for the application under consideration, where optimizing execution time is paramount for practical purposes.

For these reasons, we propose a new space-variant approach wherein pixel-dependent weights are customized to eliminate noise and streaking artifacts, remaining constant throughout all iterations of the solver.
In our variational model \eqref{eq:variational_formulation}, the space-variant regularizer  $\mathcal{R}(\x) = TV_{\w}(\x)$ is defined as: 
	\begin{align}\label{eq:TVweighted}
		TV_{\w}(\x) := \sum_{i=1}^n \w_i \sqrt{\left( \D_h \x \right)_i^2 + \left( \D_v \x \right)_i^2} = || \w \odot | \D \x | ||_1,
	\end{align}
where $\w = (\w_1, \dots, \w_n) \in \R^n$ is the vector of weights, and $\odot$ is the element-wise product.
\subsection{Uniqueness of the solution}\label{ssec:unicita}

In this Section we establish  the existence and uniqueness of the solution of problem \eqref{eq:variational_formulation} when $\mathcal{F}(\x, \y^\delta) = \frac{1}{2}||\K\x-\y^\delta||_2^2$ and $\R(x)=TV_{\w}(x)$.  To accomplish this, we initially focus on demonstrating the case where the regularization function   $\mathcal{R}(\x)$ correspnds to the TV prior defined in \eqref{eq:TV_definition}. 
In the following, as it is common in imaging science, we will always refer to $\X$ as the non-negative subspace $\X = \{\x \in \R^n; \x_i \geq \boldsymbol{0}, \: \forall i=1, \dots n\}$ of $\R^n$.\\
In this setting, proving the existence of (possibly many) solutions to the problem is trivial, as the objective function is convex on the convex feasible set $\X$. \\
On the contrary, the uniqueness of the optimal solution needs to be investigated more carefully and it passes through some preliminary propositions \cite{jorgensen2015testable}. In this Section, we denote as $\mathcal{J}(\x) := \mathcal{J}(\x, \y^\delta)$ when the explicit dependence on $\y^\delta$ is not relevant. Moreover, we denote as $\mathcal{J}^* = \min_{\x\in\X} \mathcal{J}(\x)$, whereas $\mathcal{M}:= \{ \x\in\X | \mathcal{J}(\x) = \mathcal{J}^*\}$ represents the set of the global minimizers, which is not empty since $\mathcal{J}$ is convex and coercive. We can now state the first important results.

\begin{proposition}\label{prop:disugJ}  
    For any $\x_1, \x_2 \in \R^n$, it holds:
    \begin{equation*}
        \mathcal{J} \left( \frac{\x_1 + \x_2}{2}\right) \leq \frac{1}{2} \left( \mathcal{J}(\x_1) + \mathcal{J}(\x_2)\right) - \frac{1}{8} ||\K\x_1 - \K\x_2||_2^2.
    \end{equation*}
\end{proposition}

\begin{proof}
    The proof is a simple algebraic manipulation of the objective function. Indeed, since $TV(\x) = || \D \x ||_{2, 1}$ is a norm, by the triangular inequality we get:
    \begin{align*}
        \mathcal{J} \left( \frac{\x_1 + \x_2}{2}\right) & = \frac{1}{2} \left\| \K \left( \frac{\x_1 + \x_2}{2} \right) - \y^\delta \right\|_2^2 + \lambda \mathcal{R}\left( \frac{\x_1 + \x_2}{2}\right) \\
        &
        \leq \frac{1}{2} \left\| \K \left( \frac{\x_1 + \x_2}{2} \right) - \y^\delta \right\|_2^2 + \frac{\lambda}{2}\left( \mathcal{R}(\x_1) +\mathcal{R}(\x_2) \right).
    \end{align*}
    Moreover, the linearity of $\K$ gives:
    \begin{align*}
        \frac{1}{2} \left\| \K \left( \frac{\x_1 + \x_2}{2} \right) - \y^\delta \right\|_2^2 & = \frac{1}{2} \left\| \frac{\K\x_1 + \K\x_2 - 2\y^\delta}{2} \right\|_2^2 \\ &= \frac{1}{8} \left\| \left(\K\x_1 - \y^\delta \right) + \left(\K\x_2 - \y^\delta \right) \right\|_2^2 \\ & = \frac{1}{8} \left( \left\| \K\x_1 - \y^\delta \right\|_2^2 + \left\| \K\x_2 - \y^\delta \right\|_2^2 + 2 \langle \K\x_1 - \y^\delta, \K\x_2 - \y^\delta \rangle \right).
    \end{align*}
    Note that, by the bilinearity of the scalar product:
    \begin{align*}
        \langle \K\x_1 - \y^\delta, \K\x_2 - \y^\delta \rangle = \langle \K\x_1, \K\x_2 \rangle + \langle \y^\delta, \y^\delta \rangle - \langle \K\x_1, \y^\delta \rangle - \langle \K\x_2, \y^\delta \rangle.
    \end{align*}
    Now we can complete the square of $\langle \K\x_1, \K\x_2 \rangle$, $\langle \K\x_1, \y^\delta \rangle$, and $\langle \K\x_2, \y^\delta \rangle$ by adding and subtracting $\frac{|| \K\x_1 ||_2^2}{2}$ and $\frac{|| \K\x_2 ||_2^2}{2}$ to obtain:
    \begin{align*}
        2\langle \K\x_1 - \y^\delta, \K\x_2 - \y^\delta \rangle = - || \K\x_1 - \K\x_2 ||_2^2 + || \K\x_1 - \y^\delta ||_2^2 + || \K\x_2 - \y^\delta ||_2^2.
    \end{align*}
    Substituting this value in the inequality above gives:
    \begin{align*}
        \mathcal{J} \left( \frac{\x_1 + \x_2}{2}\right) &\leq \frac{1}{2} \left\| \K \left( \frac{\x_1 + \x_2}{2} \right) - \y^\delta \right\|_2^2 + \frac{\lambda}{2}\left( \mathcal{R}(\x_1) +\mathcal{R}(\x_2) \right) \\
        &= \frac{1}{2} \left( \frac{1}{2}\left\| \K\x_1 - \y^\delta \right\|_2^2 + \frac{1}{2}\left\| \K\x_2 - \y^\delta \right\|_2^2 - \frac{1}{4}\left\| \K\x_1 - \K\x_2 \right\|_2^2 \right) + \frac{\lambda}{2}\left( \mathcal{R}(\x_1) +\mathcal{R}(\x_2) \right).
    \end{align*}
    Grouping together the fidelities with the corresponding regularization terms to form $\mathcal{J}$, we finally obtain:
    \begin{align*}
        \mathcal{J} \left( \frac{\x_1 + \x_2}{2}\right) \leq \frac{1}{2} \left( \mathcal{J}(\x_1) + \mathcal{J}(\x_2) \right) - \frac{1}{8} || \K\x_1 - \K\x_2||_2^2.
    \end{align*}
\end{proof}

An important consequence of this Proposition is that the difference of any two minimizers of $\mathcal{J}(\x)$ must lie in the kernel of $\K$.
\begin{proposition} \label{prop:x1x2InKer} 
    For any $\x_1, \x_2 \in \mathcal{M}$, it holds $\x_1-\x_2 \in \ker(\K)$.
\end{proposition}

\begin{proof}
    Since $\mathcal{J}(\x)$ is convex, the set $\mathcal{M}$ of its minima is convex. Consequently, if both $\x_1, \x_2 \in \mathcal{M}$, also $\frac{\x_1 + \x_2}{2} \in \mathcal{M}$. 
    However, by Proposition \ref{prop:disugJ}, we get:
    \begin{align*}
        \mathcal{J} \left( \frac{\x_1 + \x_2}{2}\right) & \leq \frac{1}{2} \left( \mathcal{J}(\x_1) + \mathcal{J}(\x_2) \right) - \frac{1}{8} || \K\x_1 - \K\x_2||_2^2,
    \end{align*}
    where:
    \begin{align*}
        \mathcal{J}(\x_1) = \mathcal{J}(\x_2) = \mathcal{J} \left( \frac{\x_1 + \x_2}{2}\right) = \mathcal{J}^*.
    \end{align*}
    Consequently, we get: 
     \begin{align*}
        \mathcal{J}^* \leq \frac{1}{2} \left( \mathcal{J}^* + \mathcal{J}^* \right) - \frac{1}{4} || \K\x_1 - \K\x_2||_2^2 = \mathcal{J}^* - \frac{1}{8} || \K\x_1 - \K\x_2||_2^2 
    \end{align*}
     and it holds if and only if $|| \K\x_1 - \K\x_2||_2^2 = 0$, i.e. if $\K \left( \x_1 - \x_2 \right) = \boldsymbol{0}$ which concludes the proof.   
\end{proof}

Note that, since $\K$ is undersampled by assumption, $\ker(\K)$ is not trivial in general, which implies that Proposition \ref{prop:x1x2InKer} is not sufficient to prove that the solution of \eqref{eq:variational_formulation} is unique. However, Proposition \ref{prop:x1x2InKer} implies an interesting property of the TV regularizer.
\begin{proposition} \label{prop:RUguali} 
    For any $\x_1, \x_2 \in \mathcal{M}$, it holds $\mathcal{R}(\x_1) =\mathcal{R}(\x_2)$.
\end{proposition}

\begin{proof}
    Since we assume $\x_1, \x_2 \in \mathcal{M}$, Proposition \ref{prop:x1x2InKer} suggests that $\x_1 - \x_2 \in \ker(\K)$, i.e. $\K\x_1 = \K\x_2$ and it implies that 
    $\mathcal{F}(\x_1, \y^\delta) = \mathcal{F}(\x_2, \y^\delta)$. Moreover, by convexity, $\mathcal{J}(\x_1) = \mathcal{J}(\x_2) = \mathcal{J}^*$, hence  
    the result trivially follows.
\end{proof}
A last step to prove the uniqueness of the solution to the TV-grounded Problem \eqref{eq:variational_formulation} requires fixing  some preliminary notations.
Denoting  by $I_h = \{1, \dots, n\}$ and $I_v = \{n+1, \dots, 2n\}$ the sets of indices of $\D\x$ related to the horizontal and vertical derivatives of $\x$, the operator $\hat{\D}\x: \R^n \to \R^{2n}$ is defined as follows:
\begin{align}\label{eq:Dhatx_definition1}
    \forall\ i\in I_h, \quad ( \hat{\D}\x )_i = ( \hat{\D}_h \x )_i := \begin{cases} 
        \frac{\left( \D_h\x \right)_i}{\left(| \D \x | \right)_i} & \mbox{if} \left(| \D \x | \right)_i \neq 0, \\
        \frac{1}{2} & \mbox{otherwise},
    \end{cases}
\end{align}
\begin{align}\label{eq:Dhatx_definition2}
        \forall\ i\in I_v,  \quad ( \hat{\D}\x )_i = ( \hat{\D}_v \x )_{i-n} :=\begin{cases}
        \frac{\left( \D_v\x \right)_{i-n}}{\left(| \D \x | \right)_{i-n}} & \mbox{if} \left(| \D \x | \right)_{i-n} \neq 0, \\
        \frac{1}{2} & \mbox{otherwise},
    \end{cases}
\end{align}
where $\left(| \D \x | \right)_i$ is the $i$-th element of the gradient magnitude of $\x$, as defined in \eqref{eq:gradient_magnitude}. The importance of $\hat{\D}$ resides in the following property.


\begin{proposition}\label{prop:norm21scalarprod}  
    For any $\x \in \X$, it holds $||\D\x||_{2,1} = \langle \hat{\D}\x, \D\x \rangle$.
\end{proposition}

\begin{proof}
    First of all, by definition, we get:
    \begin{align*}
        ||\D\x||_{2,1} = \sum_{i=1}^n \left(| \D\x |\right)_i.
    \end{align*}
    whereas:
    \begin{align*}
        \langle \hat{\D}\x, \D\x \rangle = \sum_{i=1}^{2n} (\hat{\D}\x )_i \left(\D\x \right)_i.
    \end{align*}
    Note that, for any $i \in I_h$ such that $\left(|\D\x\right|)_i \neq 0$ it holds $(\hat{\D}\x )_i \left(\D\x \right)_i = \frac{\left(\D\x \right)_i^2}{\left(| \D \x |\right)_i}$, and similarly $(\hat{\D}\x )_i \left(\D\x \right)_i = \frac{\left(\D\x \right)_i^2}{\left(| \D \x |\right)_{i-n}}$ for any $i \in I_v$ such that $\left(|\D\x\right|)_{i-n} \neq 0$.
    On the contrary, $(\hat{\D}\x )_i \left(\D\x \right)_i = 0$ whenever $\left(|\D\x\right|)_i = 0$ for $i \in I_h$ and whenever $\left(|\D\x\right|)_{i-n} = 0$ for $i \in I_v$. 
    Consequently:
    \begin{align*}
        \sum_{i=1}^{2n} (\hat{\D}\x )_i \left(\D\x \right)_i & = \sum_{i: \left(|\D\x\right|)_i \neq 0} \frac{\left(\D\x \right)_i^2}{\left(| \D \x |\right)_i} =  \sum_{\substack{i \in I_h \\ \left(|\D\x\right|)_i \neq 0}} \frac{\left(\D_h\x \right)_i^2}{\left(| \D \x |\right)_i} + \sum_{\substack{i \in I_v \\ \left(|\D\x\right|)_{i-n} \neq 0}} \frac{\left(\D_v\x \right)_{i-n}^2}{\left(| \D \x |\right)_{i-n}}
        \\
        &
        = \sum_{\substack{i \in I_h \\ \left(|\D\x\right|)_i \neq 0}} \frac{\left(\D_h\x \right)_i^2 + \left(\D_v\x \right)_i^2}{\left(| \D \x |\right)_i} 
        = \sum_{\substack{i \in I_h \\ \left(|\D\x\right|)_i \neq 0}} \frac{\left( | \D\x | \right)_i^2}{\left(| \D \x |\right)_i} + \sum_{\substack{i \in I_h \\ \left(|\D\x\right|)_i = 0}} \left(|\D\x\right|)_i \\ 
        &
        = \sum_{i=1}^n \left(|\D\x\right|)_i = || \D\x ||_{2, 1},
    \end{align*}
    which concludes the proof.
\end{proof}

With these tools, we can now prove the last preliminary propositions, which will allow to prove the main result for this Section.
\begin{proposition}\label{prop:S1_vectorialspace}  
    Let $\x_1\in\mathcal{M}$ and $\bar{I} = \{ i \in \{1, \dots, n\} \ | \ \left(| \D \x_1 |\right)_i = 0 \}$. Considering the set $S_1 := \{ \boldsymbol{v} \in \R^n \: | \:  \ \left(| \D \boldsymbol{v}|\right)_i = 0\ \forall\ i\in\bar{I} \}$. Then $S_1$ is a vectorial space and $\x_1\in S_1$.
\end{proposition}

\begin{proof}
    Clearly, $\x_1 \in S_1$ since $\left(| \D \x_1|\right)_i = 0$ for any $i \in \bar{I}$. Moreover, $\boldsymbol{0} \in S_1$ since $\left(| \D \boldsymbol{0}|\right)_i$ equals $0$ for any $i = 1, \dots, n$ and in particular for any $i \in \bar{I}$. Now, let $\boldsymbol{v}, \boldsymbol{w} \in S_1$ and $\alpha, \beta \in \R$. Then $\boldsymbol{t} := \alpha\boldsymbol{v} + \beta \boldsymbol{w} \in S_1$ since $\left(| \D \boldsymbol{t}|\right)_i = \alpha \left(| \D \boldsymbol{v}|\right)_i + \beta \left(| \D \boldsymbol{w}|\right)_i = 0$ for any $i \in \bar{I}$. This concludes the proof. 
\end{proof}

We can now prove the uniqueness of the solution of an underdetermined inverse problem with TV regularization, with the following Theorem.
\begin{theorem}[Unicity theorem]\label{theo:unicita}
    Let $\x_1\in\mathcal{M}$ and assume:
    \begin{enumerate}[label=(\roman*)]
        \item $\exists \ \z \in \R^n$ s.t. $\D^T \hat{\D} \x_1 = \K^T \z$, 
        \item $\ker(\K) \cap S_1 = \{ \boldsymbol{0} \}$.
    \end{enumerate}
    Then the solution of Problem \eqref{eq:variational_formulation} with $\mathcal{F}(\x, \y^\delta) = \frac{1}{2}||\K\x-\y^\delta||_2^2$ and $\R(x)=TV(x)$ is unique, i.e.  $\mathcal{M} = \{\x_1\}$.
\end{theorem}
\begin{proof}
    Let $\x_1 \in \mathcal{M}$ and we consider $\x_2 \in \mathcal{M}$ to show that necessarily $\x_1 = \x_2$. To this aim, we consider the two complementary cases of $\x_2 \in S_1$ and $\x_2 \notin S_1$. 
    
    If $\x_2 \in S_1$, then $\x_1 - \x_2 \in S_1$ as well, since  $S_1$ is a vectorial space by Proposition \ref{prop:S1_vectorialspace} and it is closed by sum. Moreover, by Proposition \ref{prop:x1x2InKer}, $\x_1 - \x_2 \in \ker(\K)$. Consequently, for the second assumption of this Theorem, $\x_1 = \x_2$. 
    
    If $\x_2 \notin S_1$, there must exist at least an index $i_0 \in \bar{I}$ such that  $\left( | \D\x_2 | \right)_{i_0} \neq 0$. In addition, according to the definition of $\hat{\D}$ given in \eqref{eq:Dhatx_definition1} and \eqref{eq:Dhatx_definition2}, $\left( \D_h \x_1 \right)_{i_0} = \left( \D_v \x_1 \right)_{i_0} = 0$ and $(\hat{\D}_h \x_1)_{i_0} = (\hat{\D}_v \x_1)_{i_0} = \frac{1}{2}$, . \\ 
    Now, recalling the result in Proposition \ref{prop:norm21scalarprod}:
    \begin{align*}
        \mathcal{R}(\x_1) = || \D\x_1 ||_{2,1} = \langle \hat{\D}\x_1, \D\x_1 \rangle,
    \end{align*}
    and by means of adjoint matrix, if $\z \in \R^n$ is the variable defined by the first hypothesis:
    \begin{align*}
        \langle \hat{\D}\x_1, \D\x_1 \rangle = \langle \D^T \hat{\D}\x_1, \x_1 \rangle = \langle \K^T \z, \x_1 \rangle. 
    \end{align*}
    Thus, by Propositions \ref{prop:x1x2InKer} and \ref{prop:norm21scalarprod}, we get:
    \begin{align*}
        \mathcal{R}(\x_1) &= \langle \K^T \z, \x_1 \rangle = \langle \z, \K \x_1 \rangle = \langle \z, \K\x_2 \rangle = \langle \K^T \z, \x_2 \rangle \\ &= \langle \D^T\hat{\D}\x_1, \x_2 \rangle = \langle \hat{\D}\x_1, \D\x_2 \rangle.
    \end{align*}
    In addition, we can derive the following inequality:
    \begin{align*}
        \langle \hat{\D}\x_1, \D\x_2 \rangle & = \sum_{i=1}^{2n} (\hat{\D}\x_1)_i (\D\x_2 )_i \leq \sum_{i=1}^n \sqrt{(\hat{\D}_h\x_1)_i^2 + (\hat{\D}_v\x_1)_i^2} \cdot \sqrt{(\D_h\x_2)_i^2 + (\D_v\x_2)_i^2} \\ &= \sum_{i=1}^n \sqrt{(\hat{\D}_h\x_1)_i^2 + (\hat{\D}_v\x_1)_i^2} \cdot\left( | \D\x_2  | \right)_i,
    \end{align*}
    for the Cauchy-Schwartz inequality applied onto the 2-dimensional vectors $[(\hat{\D}_h\x_1)_i, (\hat{\D}_v\x_1)_i]$ and $[(\D_h\x_2)_i, (\D_v\x_2)_i]$, for all the indexes $i$ separately.
    To conclude, note that by definition of $\hat{\D}$, it holds:
    \begin{align}
        \sqrt{(\hat{\D}_h\x_1)_i^2 + (\hat{\D}_v\x_1)_i^2} = \begin{cases}
            1 \quad & \mbox{if } \left(| \D\x_1 | \right)_i \neq 0, \\
            \frac{1}{\sqrt{2}} \quad & \mbox{if } \left(| \D\x_1 | \right)_i = 0.
        \end{cases}
    \end{align}
    Consequently,
    \begin{align*}
        &
        \sum_{i=1}^n \sqrt{(\hat{\D}_h\x_1)_i^2 + (\hat{\D}_v\x_1)_i^2} \left( | \D\x_2  | \right)_i \\
        &
        = \sum_{i\in\{1,\ldots,n\}\setminus\{i_0\}} \sqrt{(\hat{\D}_h\x_1)_i^2 + (\hat{\D}_v\x_1)_i^2} \left( | \D\x_2  | \right)_i + \sqrt{(\hat{\D}_h\x_1)_{i_0}^2 + (\hat{\D}_v\x_1)_{i_0}^2} \left( | \D\x_2  | \right)_{i_0} \\ 
        &
        \leq \sum_{i\in\{1,\ldots,n\}\setminus\{i_0\}} \left( | \D\x_2  | \right)_i + \sqrt{(\hat{\D}_h\x_1)_{i_0}^2 + (\hat{\D}_v\x_1)_{i_0}^2} \left( | \D\x_2  | \right)_{i_0} \\ 
        &
        < \sum_{i\in\{1,\ldots,n\}\setminus\{i_0\}} \left( | \D\x_2  | \right)_i + \left( | \D\x_2  | \right)_{i_0} = \sum_{i = 1}^n \left( | \D\x_2  | \right)_i = || \D\x_2 ||_{2, 1} = \mathcal{R}(\x_2).
    \end{align*}
    Bringing together the previous inequalities, we have demonstrated that $\mathcal{R}(\x_1) < \mathcal{R}(\x_2)$, which is a contradiction with Proposition \ref{prop:RUguali}, where we proved that if $\x_1, \x_2 \in \mathcal{M}$, then $\mathcal{R}(\x_1) = \mathcal{R}(\x_2)$. Thus, every element of $\mathcal{M}$ must be in $S_1$ and it implies that $\mathcal{M} = \{ \x_1 \}$.
\end{proof}
Note that the first condition of Theorem \ref{theo:unicita} is satisfied whenever $\D^T \hat{\D} \x_1 \in Rg(\K^T)$, the range of $\K^T$. We recall that $Rg(\K^T) = \ker(\K)^\perp$, whose dimension depends on the number of CT acquisitions performed by $\K$. Consequently, Theorem \ref{theo:unicita} implies that, whenever the number of acquisitions is sufficiently large, then the solution to Problem \eqref{eq:variational_formulation} is unique. \\
In addition, Theorem \ref{theo:unicita} can be simply extended to prove the uniqueness of the solution of any space-variant TV regularized model, i.e. when $\R(x)=TV_{\w}(x)$.  
To achieve  this, consider the matrix $\W$ defined as a $2n \times 2n$ diagonal matrix with two equal diagonal blocks of size $n \times n$, where the elements $\left( \w \right)_i$ reside on the diagonal. Observing  that: 
\begin{align*}
    || \w \odot | \D \x | ||_1= || \W \D \x||_{2,1},
\end{align*}
we can state the following Corollary.
\begin{corollary}
    Replacing $\D$ with $\W \D$ in the assumptions of Theorem \ref{theo:unicita}, then the problem:
    \begin{align}\label{eq:MinWeighted}
        \min_{\x \in \X} \mathcal{J}(\x, \y^\delta) := \frac{1}{2}|| \K\x - \y^\delta ||_2^2 + \lambda  || \w \odot | \D \x | ||_1
    \end{align}
    always admits a unique solution in $\X = \{\x \in \R^n; \x_i \geq \boldsymbol{0} \ \forall i=1, \dots n\}$.
\end{corollary}
The proof is trivial by replacing $\D$ with $\W \D$ in the proof of Theorem \ref{theo:unicita}. \\
In the following, we assume that these hypotheses are satisfied by our operator $\K$.

\subsection{On the choice of the weights}\label{ssec:conv}

In this Section we focus on the selection of the weights $\w$ in the TV model \eqref{eq:MinWeighted}, which is a matter of high delicacy.\\
Ideally, $w_i$ should be small if the corresponding pixel $\x_i$ of image to be restored lies either close to an edge of a low-contrast region or on a small detail, so that the regularization strength on that pixel is moderate and the fidelity term recovers the information provided in the data $\y^\delta$. Similarly, to maximize the regularization effect, $\w_i$ should be high on large and uniform regions of $\x$ so that the emerging of streaking artifacts and the noise propagation get removed or hindered.\\
To determine the weights $\w$, we drew inspiration from the paper \cite{sidky2014cttpv}, where the authors successfully utilize the TpV regularization \eqref{eq:TpV} in few-view CT. This regularization method efficiently promotes sparsity in the gradient domain and mitigates the undesired blurring effects often associated with Total Variation regularization, particularly on critical details and edges.  In particular, we establish the weighting factors motivated by the iterative reweighting $\ell_1$-norm strategy \cite{candes2008enhancing,daubechies2010iteratively}, commonly utilized in solving TpV problems, as:
	\begin{align}\label{eq:OurWeights}
		\w_\eta(\tilde{\x}) :=  \left( \frac{\eta}{\sqrt{\eta^2 + | \D \tilde{\x} |^2}} \right)^{1 - p},
	\end{align}
where $\Tilde \x \in \R^n$ is a suitable image and $\eta\ge 0$ is a scalar parameter.

The following statements serve to elucidate the conceptual underpinnings of our weight parameters, formally.

\begin{proposition}\label{prop:w_is_a_scale_term}
    For any $\eta > 0$, any $0 < p < 1$, and any $\tilde{\x} \in \R^n$, 
    $\left(w_\eta(\tilde{\x})\right)_i \in (0, 1]$,  $\forall i = 1, \dots, n$. Moreover, $\left(\w_\eta(\tilde{\x})\right)_i = 1$ on a pixel $i \in 1, \dots, n$ if and only if $| \D \tilde{\x} |_i = 0$.
\end{proposition}

\begin{proof}
    Note that, for any $0 < p < 1$, $1 - p > 0 $. Consequently, 
    \begin{align*}
        \left(\w_\eta(\tilde{\x})\right)_i = \left( \frac{\eta}{\sqrt{\eta^2 + | \D \tilde{\x} |^2}} \right)^{1 - p} \in (\boldsymbol{0}, \boldsymbol{1}] \iff \frac{\eta}{\sqrt{\eta^2 + | \D \tilde{\x} |^2}} \in (\boldsymbol{0}, \boldsymbol{1}].
    \end{align*}
    Moreover, since both $\eta$ and $\sqrt{\eta^2 + | \D \tilde{\x} |^2}$ are strictly positive, their ratio is necessarily strictly positive. On the other side,
    \begin{align*}
        \frac{\eta}{\sqrt{\eta^2 + | \D \tilde{\x} |^2}} \leq \boldsymbol{1} \iff \eta \leq \sqrt{\eta^2 + | \D \tilde{\x} |^2} \iff | \D \tilde{\x} |^2 \geq \boldsymbol{0},
    \end{align*}
    which holds for any $\tilde{\x} \in \R^n$, and $| \D \tilde{\x} |_i^2 = 0$ on a pixel $i$ if and only if $| \D \tilde{\x} |_i=0$.
\end{proof}
Now we observe that whenever $| \D \tilde{\x} |_i = 0$ (i.e., if $\tilde{\x}_i$ is on a flat region), then $\left(\w_\eta(\tilde{\x})\right)_i = 1$, and the Total Variation is exactly weighted by $\lambda$. Conversely, whenever $| \D \tilde{\x} |_i \gg 0$ (indicating that pixel $i$ is either on an edge or a small detail), then $\left(\w_\eta(\tilde{\x})\right)_i <1$, signifying that the Total Variation regularization is weaker, and the details are preserved.\\
Focusing now on the role of $\eta$ in the computation of the weights $\w_\eta$, Figure \ref{fig:eta_plot} shows the value of $\left(\w_\eta(\tilde{\x})\right)_i$  at increasing values of $| \D \tilde{\x} |_i$, for different values of $\eta > 0$. 
It is evident that, for very small values of $\eta$, $\left(\w_\eta(\tilde{\x})\right)_i$ rapidly decreases to zero (i.e., the solution of the problem is not regularized around pixel $i$). On the other side, for large values of $\eta$, $\left(\w_\eta(\tilde{\x})\right)_i$ decreases very slowly as a function of $ | \D \tilde{\x} |_i$, which implies that small details will be flattened out.
Wrapping up the analysis of weights, we emphasize that since the optimal values of $\w_\eta$ are linked to the contrasts and the arrangement of details in $\x^{GT}$, we strive for $|\D\tilde{\x}|$ to closely resemble $|\D\x^{GT}|$, especially near the edges.

\begin{figure}
    \centering
    \includegraphics[width=0.6\linewidth]{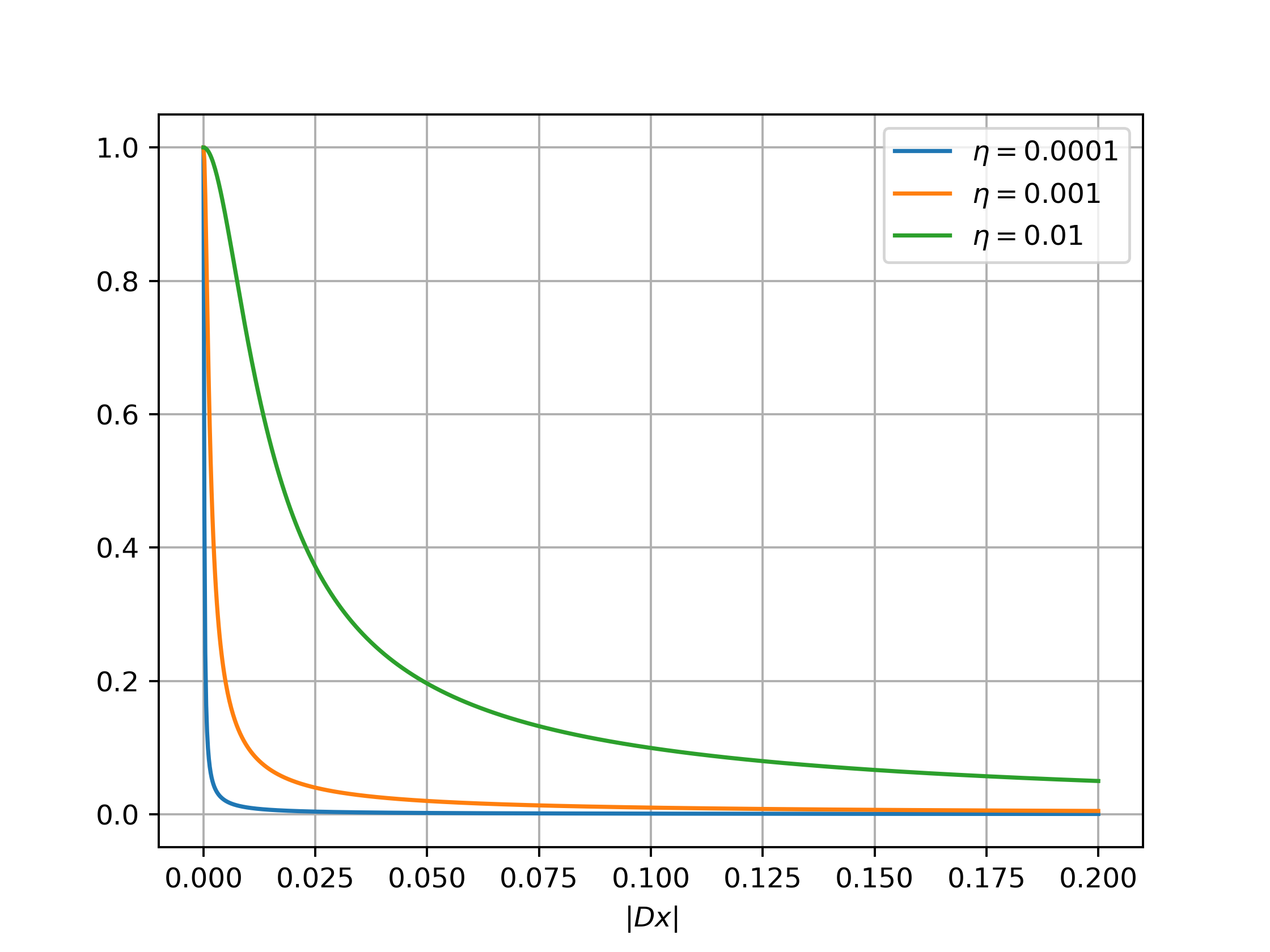}
    \caption{A plot of $\left(\w_\eta(\tilde{\x})\right)_i$ for different values of $\eta$, over $ | \D \tilde{\x} |_i$.}
    \label{fig:eta_plot}
\end{figure}

To provide a more precise specification on how to derive $\Tilde \x$ from $\y^{\delta}$  we consider $\Psi: \mathbb{R}^m \rightarrow \mathbb{R}^n$ as a Lipschitz-continuous function that maps $\y^{\delta}$ to an approximate reconstruction $\tilde{\x} = \Psi(\y^{\delta})$ of $\x^{GT}$. In accordance with the nomenclature introduced in \cite{evangelista2022or}, we denote such mappings as {\it reconstructors}.\\
Hence, given a reconstructor $\Psi: \R^m \to \R^n$, $\lambda>0$ and $\eta > 0$, we can formalize the proposed weighted TV model as:
    \begin{align}\label{eq:Psi_WL1_formulation}
         \arg\min_{\x \in \X} \frac{1}{2}|| \K\x - \y^\delta ||_2^2 + \lambda || \w_\eta(\tilde{\x}) \odot | \D \x | ||_1 
    \end{align}
    where $\tilde{\x} = \Psi(\y^\delta)$, the space-variant weights are computed as in Equation \eqref{eq:OurWeights} and $\X = \{\x \in \R^n; \ \x_i \geq \boldsymbol{0}, \ \forall i=1, \dots n\}$, representing the non-negative subspace of $\R^n$.
In the following, we denote the solution of \eqref{eq:Psi_WL1_formulation} as $\x^*_{\Psi, \delta}$ and 
we refer to this model as $\Psi$-W$\ell_1$ to remark its reliance on $\Psi$ in the weighted $\ell_1$-norm regularization (as further elaborated upon in Section \ref{ssec:PsiAlgo}). 


We have already intuitively observed that $\x^{GT}$ would be a suitable candidate for $\tilde{\x}$. This intuition is further confirmed by the numerical results obtained from simulations (presented in Section \ref{sec:numexp}). Indeed, if we use the ground truth image to compute the space-variant weights, we can achieve highly accurate solutions (that will be referred to as $\x^*_{GT, \delta}$, with a minor notational adaptation).\\
However, the ground truth images are not available in real cases, hence we need to use a reconstructor working on the noisy data. 
Under a few assumptions, we can demonstrate that 
if $|\D\Psi(\y^{\delta})|$ provides a good approximation of $|\D\x^{GT}|$, then the solution obtained considering $\w_{\eta}(\Psi(\y^{\delta}))$ is a reliable approximation of $\x^*_{GT, \delta}$.
To this aim, we need to build a sequence of reconstructors $\Psi_k: \R^m \to \R^n$ indexed by $k$, and demonstrate the following theorem.
\begin{theorem}\label{teo:conv}
    Given $\delta > 0$, let $\{ \Psi_k \}_{k \in \mathbb{N}}$ be a sequence of reconstructors such that:
    \begin{align*}
        \lim_{k \to \infty} || | \D\Psi_k(\y^\delta)| - | \D\x^{GT}| ||_1 = 0, \quad \forall \y^\delta = \K\x^{GT} + \e, || \e ||_2 \leq \delta,
    \end{align*}
    then:
    \begin{align*}
        \lim_{k \to \infty} || \x^*_{\Psi_k, \delta} - \x^*_{GT, \delta}  ||_1 = 0.
    \end{align*}
\end{theorem}
The verification of this theorem necessitates numerous observations, the proofs of which, along with the theorem's proof itself, are meticulously expounded in the Appendix \ref{sec:Appendix}. This organizational choice is made to uphold the overall readability and fluidity of the entire manuscript.

At last, we observe that the above theorem leads to the following result, straightforwardly.
\begin{corollary}\label{cor:conv}
    Given $\delta > 0$, let $\{ \Psi_k \}_{k \in \mathbb{N}}$ be a sequence of reconstructor such that:
    \begin{align*}
        \lim_{k \to \infty} || \Psi_k(\y^\delta) - \x^{GT} ||_1 = 0, \quad \forall \y^\delta = \K\x^{GT} + \e, || \e ||_2 \leq \delta,
    \end{align*}
    then:
    \begin{align*}
        \lim_{k \to \infty} || \x^*_{\Psi_k, \delta} - \x^*_{GT, \delta} ||_1 = 0.
    \end{align*}
\end{corollary}

The Corollary suggests that if we can fix a reconstructor whose output is sufficiently close to $\x^{GT}$ for any noisy datum,  
the optimal images $ \x^*_{\Psi_k, \delta}$ and $\x^*_{GT, \delta}$ (computed with the respective weights) will be sufficiently close to each other.

\subsection{Proposed $\Psi$ reconstructors}\label{ssec:PsiAlgo}

We can now propose various possible choices to define the reconstructor $\Psi(\y^\delta)$ and compute  $\tilde{\x}$ from the sinogram $\y^\delta$.

Dealing with tomographic imaging from  X-ray measurements, a simple yet fast reconstructor is the Filtered Back Projection (FBP) solver \cite{kak2001principles}. As an analytical algorithm, it operates by applying a filter to the back-projected data, disregarding the inverse problem formulation of the imaging task. Even if FBP may not address all reconstruction challenges, its speed and simplicity make it a pragmatic choice for many medical imaging scenarios, particularly when balancing computational resources and reconstruction quality is essential.

A further option is setting $\Psi(\y^{\delta})$ as an early-stopped variational solver. For example, if we still consider the global TV regularization \eqref{eq:TV_definition} and stop the iterative  algorithm for the solution  of the corresponding problem \eqref{eq:variational_formulation} according to computational constraints.
While in few-view CT the FBP solution is typically characterized by neat details but also by noise and streaking artifacts, the solution computed by an early stopped TV solver is expected to be a blurred version of the ground truth image, with moderate few-view-caused artifacts.
However, we observe that these reconstructors can only provide images approximating $\x^{GT}$ in a classical interpretation.

\begin{figure}
\centering
\begin{subfigure}{0.9\textwidth}
    \includegraphics[trim=0 45mm 0 0,clip, width=\textwidth]{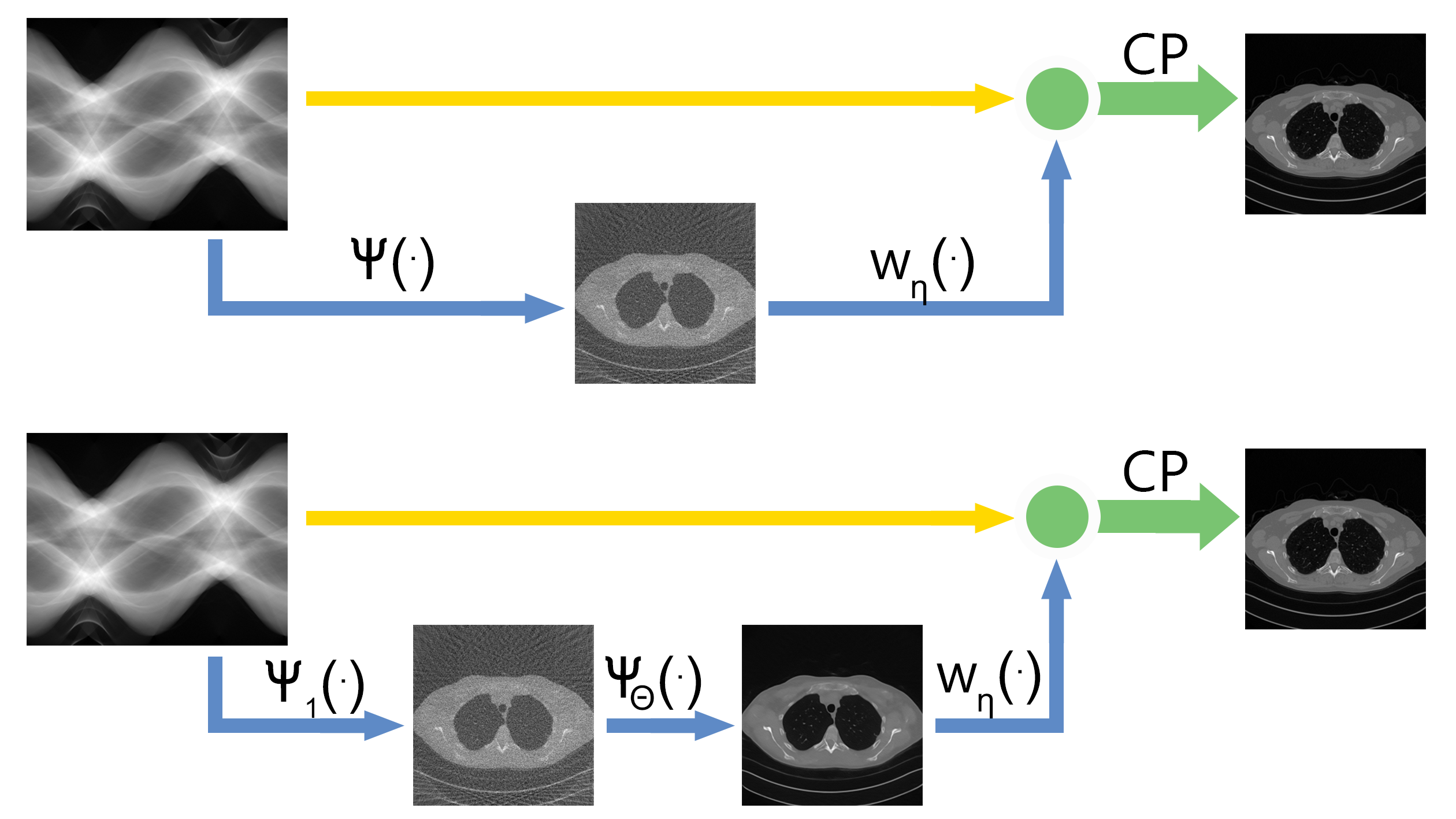}
\end{subfigure}\\ \vspace{7mm}
\begin{subfigure}{0.9\textwidth}
    \includegraphics[trim=0 0 0 45mm,clip, width=\textwidth]{fig/Psi-Wl1.png}
\end{subfigure}
\caption{Visual depictions of the proposed $\Psi$-W$\ell_1$ method, illustrating the key computational steps of our proposal, spanning from the sinogram to the final reconstruction. On the top, the scenario where $\Psi$ is implemented by a single solver; on the bottom,  the case where it is formed through a composition of two functions.}
\label{fig:GraphicalAbstract}
\end{figure}

Another method entails the use of a composition of two reconstructors. We can in fact define $\Psi(\y^{\delta})=\Psi_2( \Psi_1(\y^{\delta}))$, where $\Psi_1(\y^{\delta})$ can be one of the previously described reconstructors applied to map the data $\y^\delta$ to the image domain, whereas $\Psi_2: \mathbb{R}^n \to \mathbb{R}^n$ is designed to enhance the quality of $\Psi_1(\y^{\delta})$. 
In Figure \ref{fig:GraphicalAbstract}, we visually illustrate the contrast between the current and the previous methods for the establishing of $\Psi$ in the $\Psi$-W$\ell_1$ model. \\
In our specific case, we have chosen $\Psi_2$ to be a neural network (NN).
As extensively documented in the literature, NNs possess the capability to execute efficient image-processing computations through pre-trained deep architectures, offering the flexibility to estimate various image transformations, such as the image gradients \cite{lu2022single, cascarano2022plug} or the image shearlets coefficients \cite{bubba2019learning}. However, it is crucial to note that the accuracy of NN reconstructors is contingent upon the quality of the training samples and a considerable computational effort is required during their offline training phase.\\
It is noteworthy that, in the context of medical imaging, the collection of data from the same clinical device allows for the construction of a coherent dataset for training a network.
Therefore, we can presume the availability of the training set $\mathcal{D}$, constituted by $M$ couples $(\x^{GT}_j, \Psi_1(\y^\delta_j))$ of high-quality reconstructed images $\x^{GT}_j$ (to be used as ground truth images) and the corresponding low-quality reconstructions  $\Psi_1(\y^\delta_j)$, for all $j=1, \dots, M$.\\
We utilize in our experiments the renowned Residual U-Net architecture as described in \cite{green_post_processing,evangelista2023ambiguity}.
A neural network-based reconstructor $\Psi_2$ can be seen as a $\Theta$ parameter-dependent solver, where $\Theta$ represents the high-dimensional vector of the NN inner weights; therefore we will denote it as $\Psi_{\Theta}$ in the following. 
To set these weights, we need to train the network having established a loss function to be minimized.
Inspired by the convergence results in Theorem \ref{teo:conv} and Corollary \ref{cor:conv} where we respectively assume that $| \D\Psi_k(\y^\delta)| \to | \D\x^{GT}|$ and $\Psi_k(\y^\delta) \to \x^{GT} \ \mbox{for\ } k\to\infty$, we have designed a loss function such that each $\Psi_\Theta( \Psi_1(\y^\delta_j))$ image approximates not only the ground truth one but also the magnitude images $|\D \x^{GT}_j|$.
To this aim, inspired by the regularization technique proposed in \cite{Benfenati_2020}, we train $\Psi_\Theta$ with the \emph{elastic loss}: 
    \begin{equation}\label{eq:ElasticLoss}
        \Theta^* \in \arg\min_{\Theta} \sum_{j=1}^M \alpha || |\D \x^{GT}_j| - |\D \Psi_\Theta(\y^\delta_j)| ||_2^2 \ +  \ (1-\alpha)|| \x^{GT}_j - \Psi_\Theta(\y^\delta_j) ||_2^2
    \end{equation}
where $\alpha \in [0, 1]$.
Setting $\alpha = 0$ yields a network aimed at achieving the best image approximation, whereas $\alpha = 1$ directs the training process to focus solely on learning gradient images.
We anticipate that, in our experimental setting, we will consider $\alpha = \{0, 0.5, 1\}$ to test the effectiveness of the elastic loss in these extreme $\alpha$ values as well as in its central elastic coefficient.
In addition, we will use $\Psi_1$ as the FBP algorithm, to get a comprehensive reconstructing framework of very fast execution. 
As training samples, we use $M$=3305 couples $(\x^{GT}_j, \Psi_1(\y^\delta_j))$ of real chest images extracted from the Low Dose CT Grand Challenge data set by the Mayo Clinic \cite{mccollough2016tu}. They are $256\times 256$ pixel resolution and are normalized in the unitary interval.  
We perform 50 epochs with  Adam optimizer with a learning rate equal to $10^{-3}$ and $\beta_1=0.9$, $\beta_2=0.9999$ to complete the training phase. 

\section{Chambolle-Pock algorithm}\label{sec:Algo}


To solve the optimization problem \eqref{eq:Psi_WL1_formulation}, we consider the Chambolle-Pock (CP) algorithm \cite{chambolle2011first}, that is a popular iterative method used to solve the $\text{TV}$-regularized inverse problem \cite{piccolomini2021model,sidky2014cttpv} in few-view CT. In its original formulation, it can be employed to minimize an objective function of the form:
\begin{align}\label{eq:CP_general_formulation}
    \min_{\x \in \R^n} F(\M\x) + G(\x),
\end{align}
where both $F$ and $G$ are real-valued, proper, convex, lower semi-continuous functions  and $\M$ is a linear operator from $\R ^n$ to $\R^s$. 
Note that there are no constraints on the smoothness of either $F$ and $G$; therefore, the method can be applied to our problem by setting:
\begin{align}\label{eq:CP_funtions}
    \begin{cases}
        G(\x) = \iota_{\X}(\x), \\
        F(\M\x) = \mathcal{J}_{\Psi}(\x, \y^\delta) = \frac{1}{2} || \K\x - \y^\delta ||_2^2 + \lambda || \w_\eta(\tilde{\x}) \odot | \D \x | ||_1,
    \end{cases}
\end{align}
where $\iota_{\X}(\x)$ is the indicator function of the feasible set $\X$. As already stated, we consider $\X$ as the  non-negative subspace of $\R^n$.\\
To apply the CP method to our problem, we  define the linear operator $\M \in \R^{s \times n}$  by concatenating row-wise $\K$ and $\D$, namely $\M = \left[ \K; \D \right]$ and $s=(m + 2n)$. The CP algorithm considers the primal-dual formulation of \eqref{eq:CP_general_formulation}, that reads:
\begin{align}
    \min_{\x \in \R^n} \max_{\z \in \R^{m+2n} } \z^T \M \x + G(\x) - F^*(\z),
\end{align}
where $F^*$ is the convex conjugate of $F$ \cite{bauschke2017correction}, defined as:
\begin{align}
    F^*(\z^*) := \sup_{\z \in \R^{m+2n} } \left\{ \z^T \z^* - F(\z) \right\}.
\end{align} 
Given a starting guess for both the primal variable $\x^{(0)}$ and the dual variable $\z^{(0)}$, the update rule is the following:
\begin{align}\label{eq:CP_iterates}
    \begin{cases}
        \z^{(k+1)} = \prox_{\sigma F^*}\left(\z^{(k)} + \sigma \M\bar{\x}^{(k)}\right), \\
        \x^{(k+1)} = \prox_{\tau G}\left(\x^{(k)} - \tau \M^T\z^{(k+1)}\right), \\
        \bar{\x}^{(k+1)} = \x^{(k+1)} + \beta \left(\x^{(k+1)} - \x^{(k)} \right),
    \end{cases}
\end{align}
where $\bar{\x}^{(0)} = \boldsymbol{0}$, $\beta \in [0, 1]$ is a parameter that we set equal to 1 in the experiments, while $\sigma > 0$, $\tau > 0$ are computed as $\sigma = \tau \approx \frac{1}{|| \boldsymbol{M} ||_2}$. A reliable approximation of $|| \boldsymbol{M} ||_2$ can be computed by means of the power method  \cite{epperson2021introduction}. 

Focusing on $G(\x)$, its proximal operator corresponds to the projection $\mathcal{P}_+$ over the non-negative subspace $\X$. Consequently the updating rule of the primal variable $\x^{(k)}$ becomes:
\begin{align}
    \x^{(k+1)} = \mathcal{P}_+\left(\x^{(k)} - \tau \M^T\z^{(k+1)}\right).
\end{align}
The explicit derivation of $\prox_{\sigma F^*}$ requires the introduction of two dual variables, $\p \in \R^m$ and $\q \in \R^{2n}$, such that:
\begin{align}
    F(\p, \q) = \underbrace{\frac{1}{2} || \p - \y^\delta ||_2^2}_{:= F_1(\p)} + \underbrace{\lambda || \w_\eta(\tilde{\x}) \odot | \q | ||_1}_{:= F_2(\q)}.
\end{align}
Denominating $F_1(\p) = \frac{1}{2} || \p - \y^\delta ||_2^2$, its convex conjugate $F_1^*(\p)$ can be easily computed as:
\begin{align}
    F_1^*(\p) &= \sup_{\p'} \left\{ \p^T \p' - \frac{1}{2} || \p' - \y^\delta ||_2^2 \right\} \\ &= \p^T \y^\delta + \frac{3}{2} || \p ||_2^2,
\end{align}
since the problem defining $F_1^*$ is quadratic in $\p'$ and it can be solved by imposing the optimality condition of null gradient. 
Consequently, the proximal map of $F_1^*$ gets:
\begin{align}
    \prox_{\sigma F_1^*}(\p) &= \arg\min_{\p'} (\p')^T \y^\delta + \frac{3}{2} || \p' ||_2^2 + \frac{1}{2\sigma} || \p' - \p ||_2^2 \\&= \frac{\p - \sigma \y^\delta}{1 + 3\sigma}.
\end{align}
Similarly, calling $F_2(\q) = \lambda || \w_\eta(\tilde{\x}) \odot | \q | ||_1$, we derive:
\begin{align}
    F_2^*(\q) = \begin{cases}
        0 \quad &\text{if } \q \leq \lambda \w_\eta(\tilde{\x}), \\
        \infty \quad &\text{otherwise}, 
    \end{cases}
\end{align}
which implies that:
\begin{align}
    \prox_{\sigma F_2^*}(\q) &= \arg\min_{\q'} F_2^*(\q') + \frac{1}{2\sigma} || \q - \q' ||_2^2 \\ &= \frac{\lambda \w_\eta(\tilde{\x}) \odot \q}{\max \left( \lambda \w_\eta(\tilde{\x}), \q \right)},
\end{align}
where both the maximum and the division are taken element-wise. More details on the computation of the proximal operator of $F_2^*$ can be found in \cite{sidky2014cttpv}. \\

Plugging these results into the CP scheme \eqref{eq:CP_iterates} with $\z^{(k)} = (\p^{(k)}, \q^{(k)})$ leads to the following updating rules:
\begin{align}
    \begin{cases}
        \p^{(k+1)} = \frac{\p^{(k)} + \sigma \left(\K \bar{\x}^{(k)} - \y^\delta \right)}{1 + 3\sigma}, \\
        \q^{(k+1)} = \frac{\lambda \w_\eta(\tilde{\x}) \odot \left(\q^{(k)} + \sigma | \D \bar{\x}^{(k)} | \right)}{\max \left( \lambda \w_\eta(\tilde{\x}), \q^{(k)} + \sigma | \D \bar{\x}^{(k)} | \right)}, \\
        \x^{(k+1)} = \mathcal{P}_+\left(\x^{(k)} - \tau \M^T\begin{bmatrix}
            \p^{(k+1)} \\ \q^{(k+1)}
        \end{bmatrix} \right), \\
        \bar{\x}^{(k+1)} = \x^{(k+1)} + \beta \left(\x^{(k+1)} - \x^{(k)} \right).
    \end{cases}
\end{align}
The corresponding CP algorithm is reported in Algorithm \ref{algo1}.

\begin{algorithm}
\caption{The Chambolle-Pock algorithm to solve Problem \eqref{eq:variational_formulation}}\label{algo1}
\begin{algorithmic}[1]
\Require a linear operator $\K$, corrupted data $\y^\delta \in \R^m$, a regularization parameter $\lambda > 0$, a reconstructor $\Psi: \R^m \to \R^n$
\Require $\bar{\x}^{0} = \x^{0} \in \R^n$, $\p^{(0)} \in \R^m$, $\q^{(0)} \in \R^{2n}$
\State \textbf{define} $\M = \left[ \K; \D \right]$, $\gamma \approx || \M ||_2, \tau = \sigma = \gamma^{-1}$, $\beta = 1$, $\eta > 0$
\State \textbf{initialize} $k = 0$
\Repeat
    \State $\p^{(k+1)} = \frac{\p^{(k)} + \sigma \left(\K \bar{\x}^{(k)} - \y^\delta \right)}{1 + 3\sigma}$ \Comment{Update dual variables}
    \State $\q^{(k+1)} = \frac{\lambda \w_\eta(\tilde{\x}) \odot \left(\q^{(k)} + \sigma | \D \bar{\x}^{(k)} | \right)}{\max \left( \lambda \w_\eta(\tilde{\x}), \q^{(k)} + \sigma | \D \bar{\x}^{(k)} | \right)}$
    \State
    \State $\x^{(k+1)} = \mathcal{P}_+\left(\x^{(k)} - \tau \M^T\begin{bmatrix}
            \p^{(k+1)} \\ \q^{(k+1)}
        \end{bmatrix} \right)$ \Comment{Update primal variable}
    \State
    \State $\bar{\x}^{(k+1)}=\x^{(k+1)}+\beta(\x^{(k+1)}-\x^{(k)})$ \Comment{Update inertia term}
    \State
    \State \textbf{update } $PDG(\x^{(k+1)}, \p^{(k+1)}, \q^{(k+1)})$ as in \eqref{eq:primal_dual_gap}
    \State \textbf{update } $k = k + 1$
\Until{$PDG(\x^{(k)}, \p^{(k)}, \q^{(k)}) \leq \epsilon_{\mathcal{J}}$ \textbf{or} $|| \x^{(k+1)} - \x^{(k)} ||_2 \leq \epsilon_{\x} || \x^{(k)} ||_2$}
\end{algorithmic}
\end{algorithm}

We stop the CP iterations according to a combination of two convergence criteria: firstly, at any iteration, we check whether $|| \x^{(k+1)} - \x^{(k)} ||_2 \leq \epsilon_{\x} || \x^{(k)} ||_2$, to stop the algorithm if the method get stucked on a flat region of the objective function. Secondly, as the Chambolle-Pock algorithm is a primal-dual method, we consider a primal-dual gap (PDG) criterium, terminating the iterations if:
\begin{align}\label{eq:primal_dual_gap}
    \underbrace{F(\M\x^{(k)}) + G(\x^{(k)}) + G^*\left(\M^T\begin{bmatrix}
            \p^{(k)} \\ \q^{(k)}
        \end{bmatrix}\right) + F^*\left(\begin{bmatrix}
            \p^{(k)} \\ \q^{(k)}
        \end{bmatrix}\right)}_{\text{Primal-Dual Gap } PDG(\x^{(k)}, \p^{(k)}, \q^{(k)})} \leq \epsilon_{\mathcal{J}}.
\end{align}

We remark that the CP algorithm is convergent if both $F(\M\x)$ and $G(\x)$ are convex, proper and lower semi-continuous, with a theoretical convergence rate of $\mathcal{O}(\frac{1}{k^2})$ \cite{chambolle2011first}, and we fit these requirements.

\section{Numerical experiments}
\label{sec:numexp}

In this Section, we present the numerical results obtained on some synthetic and real test images.
The tests on synthetic images primarily aim to assess the potential of the model and analyze the method's behavior in an optimal scenario. Real images, as such, inherently contain artifacts and some noise, which may not necessarily be desirable to reproduce in the reconstruction. On the other hand, they possess characteristics and areas of interest that cannot be perfectly replicated in a simulation.

In both cases we compute a test problem, by setting a very sparse fan beam geometry with 45 angles in the range $[0,180]$ and computing a projection matrix $\K$ by using Astra-Toolbox \cite{astra_1,astra_2} routines. 
The sinogram $\y^{\delta}$ is then computed from the ground truth image $\x^{GT}$, according to the image formation model in Equation \eqref{eq:inverse_problem}. The noise $\e$ is obtained by sampling $\z \sim \mathcal{N}(\boldsymbol{0}, \I)$ and computing:
\begin{align*}
    \e = \nu \frac{\z}{|| \z ||_2} || \y ||_2,
\end{align*}
where we indicated with $\y$ the noiseless sinogram (i.e. with $\nu = 0$). Note that this formulation is equivalent to \eqref{eq:inverse_problem}, with $\delta = \nu || \y ||_2$.

We analyze the considered reconstructions by means of metrics such as the Relative Error (RE), the Peak Signal to Noise Ratio (PSNR) and the Structural Similarity Index (SSIM) defined in \cite{wang2004image},  with respect to the ground truth image. 

To justify the choice of the space-variant parameter, we compare the results obtained from the proposed method with the global  Total Variation regularization. In this case, the parameter $\lambda$ has been chosen heuristically to minimize the relative error.
Moreover, we conduct a comparative analysis of the outcomes achieved by varying the choice of the reconstructor $\Psi$, as discussed in Section \ref{ssec:PsiAlgo}.  We also considered the case $\tilde \x=\x ^{GT}$ for comparison.
The methods are labeled as FBP-W$\ell_1$, TV-W$\ell_1$, NN-W$\ell_1$, and GT-W$\ell_1$.
Notably, in the case of TV-W$\ell_1$ method, the CP method has been used to implement the $\Psi$ reconstructor, prematurely terminated after 100 iterations. In the case of NN-W$\ell_1$, we have used $\alpha \in \{0, 0.5, 1 \}$ in the elastic loss. \\ 

\subsection{Experiments on a synthetic image}\label{ssec:Sintetica}

The synthetic image, depicted in Figure \ref{fig:results_Sint005}, contains objects of interest in tomography: homogeneous areas of regular shape simulating tumoral masses, high-density areas simulating bones or metal, and objects with very thin edges.\\
All considered methods were stopped at 10000 iterations, after ensuring that they satisfy the convergence criteria \eqref{eq:primal_dual_gap} with $\epsilon_{\mathcal{J}}=\epsilon_{\x}=10^{-5}$. 
Since the image to be reconstructed has an extremely sparse gradient, the parameter $\eta$ must be small and we have set $\eta=2 \cdot 10^{-5}$ for all the methods.\\
We conducted two simulations: in the first instance, low noise was introduced with $\nu=0.005$, while in the second case, $\nu$ was set to 0.02. For both simulations, Table \ref{tab:Sintetica} presents the metrics for both the $\tilde{\x}$ image and the reconstructed image $\x^*_{\Psi, \delta}$ across the various considered approaches.\\
We first consider the low noise projections obtained with $\nu=0.005$; in this case we set $\lambda=5$ for all the methods. 
Table \ref{tab:Sintetica} highlights that the reconstructions achieved through the proposed space-variant regularization method consistently demonstrate exceptional quality, as evidenced by SSIM values surpassing 0.99. 
Notably, this high quality persists despite considerable dissimilarity among the $\tilde{\x}$ images employed for computing the weights.
Examination of RE and PSNR values reveals superior outcomes when utilizing, in sequential order, the GT-$W\ell_1$, the FBP-$W\ell_1$ and the TV-$W\ell_1$  approaches. Conversely, the global TV reconstruction yields the least favorable results across all metrics.
\begin{table}[h]
\caption{Performance results on the synthetic image with low ($\nu=0.005$) and medium ($\nu=0.02$) noise. In the first three columns the metrics relative to the image $\tilde{\x}$, in the last three columns the metrics relative to the output image $\x^*_{\Psi,\delta}$.}\label{tab:Sintetica}%
\begin{tabular*}{\textwidth}{@{\extracolsep\fill}ll rrr rrr}
\toprule
 &   & \multicolumn{3}{@{}c@{}}{ $\tilde{\x} $} & \multicolumn{3}{@{}c@{}}{$\x^*$}\\
  &         & RE & PSNR   & SSIM          & RE & PSNR   & SSIM\\
\cmidrule{3-5} \cmidrule{6-8}
\midrule
\multirow{ 4}{*}{$\nu=0.005$} 
   & GT-$W\ell_1$  & 0.0000   & 100.00  & 1.0000            & \textbf{0.0217} & \textbf{46.9881} & \textbf{0.9979}  \\
   & FBP-$W\ell_1$   & 0.5056   & 19.6281  & 0.1974           & 0.0373 & 42.2662 & 0.9953  \\
   & TV-$W\ell_1$    & \textbf{0.2236}   & \textbf{26.7167}  & \textbf{0.7335}           & 0.0621 & 37.8454 & 0.9920  \\
   & global TV    & -   & -  & -                    & 0.1236 & 31.8657 & 0.9805  \\
\midrule
\multirow{ 4}{*}{$\nu=0.02$}
   & GT-$W\ell_1$   & 0.0000   & 100.00  & 1.0000            & \textbf{0.0398} & \textbf{41.7159} & \textbf{0.9968}  \\
   & FBP-$W\ell_1$    & 0.9694   & 13.9748 & 0.0524            & 0.1025 & 33.4928 & 0.9809  \\
   & TV-$W\ell_1$   & \textbf{0.2435} & \textbf{25.9754} & \textbf{0.5968}            & 0.0911 & 34.5162 & 0.9797  \\
   & global TV    & -   & -  & -                     & 0.1249 & 31.7715 & 0.9787  \\
\bottomrule
\end{tabular*}
\end{table}

\begin{figure}
    \centering
    \includegraphics[trim={35mm 0 15mm 0},clip, width=0.3\linewidth]{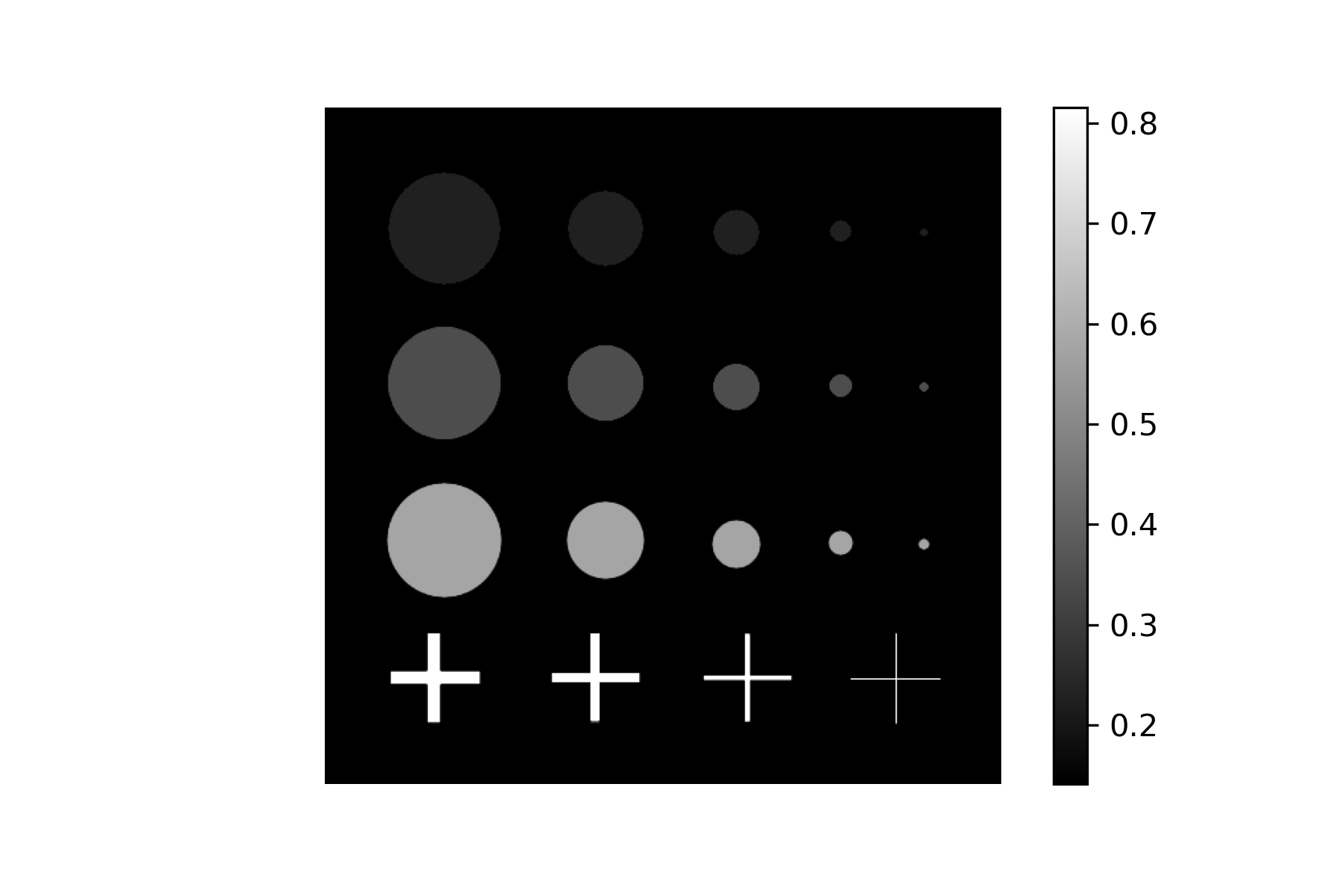}
    \includegraphics[trim={35mm 0 15mm 0},clip, width=0.3\linewidth]{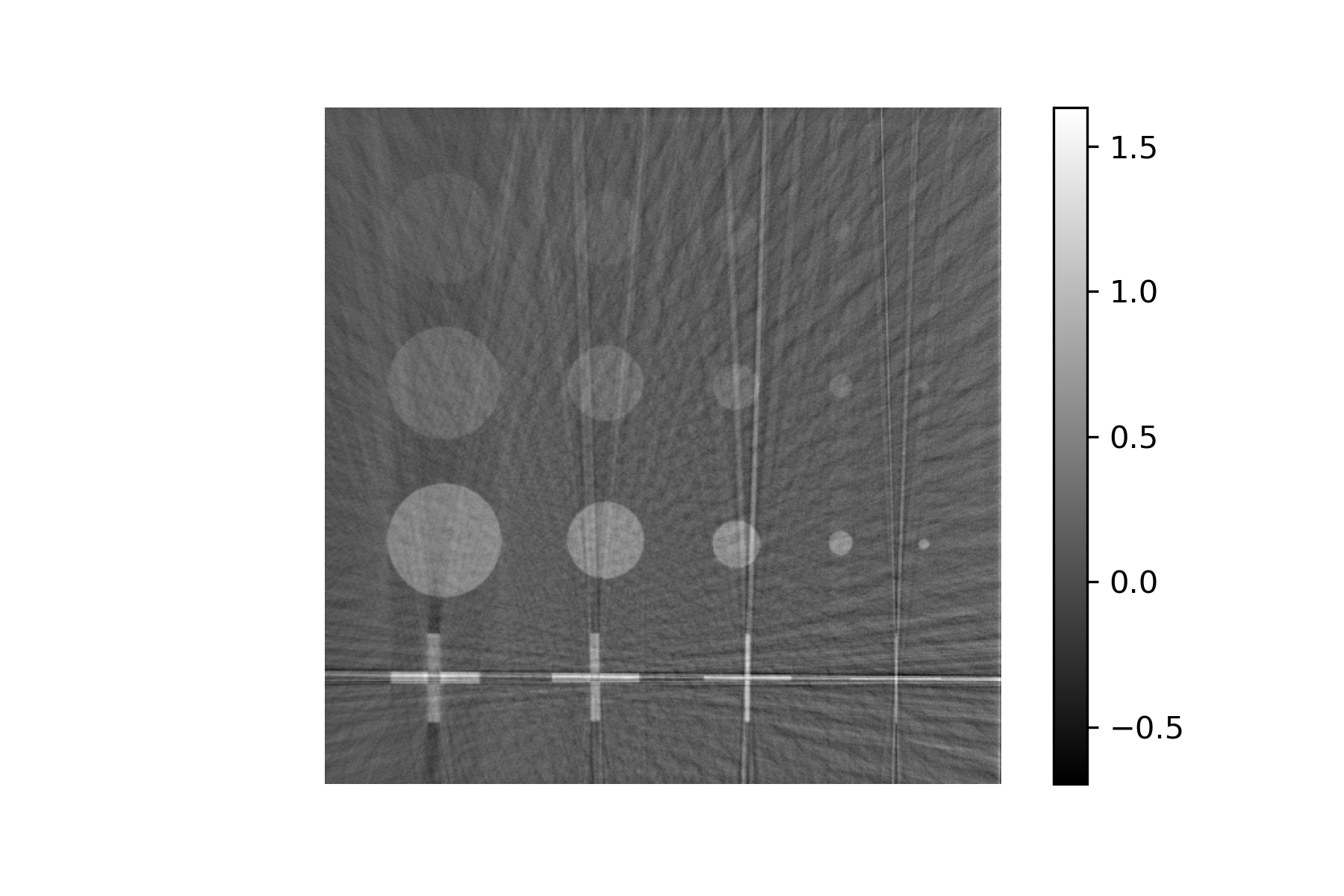}
    \includegraphics[trim={35mm 0 15mm 0},clip, width=0.3\linewidth]{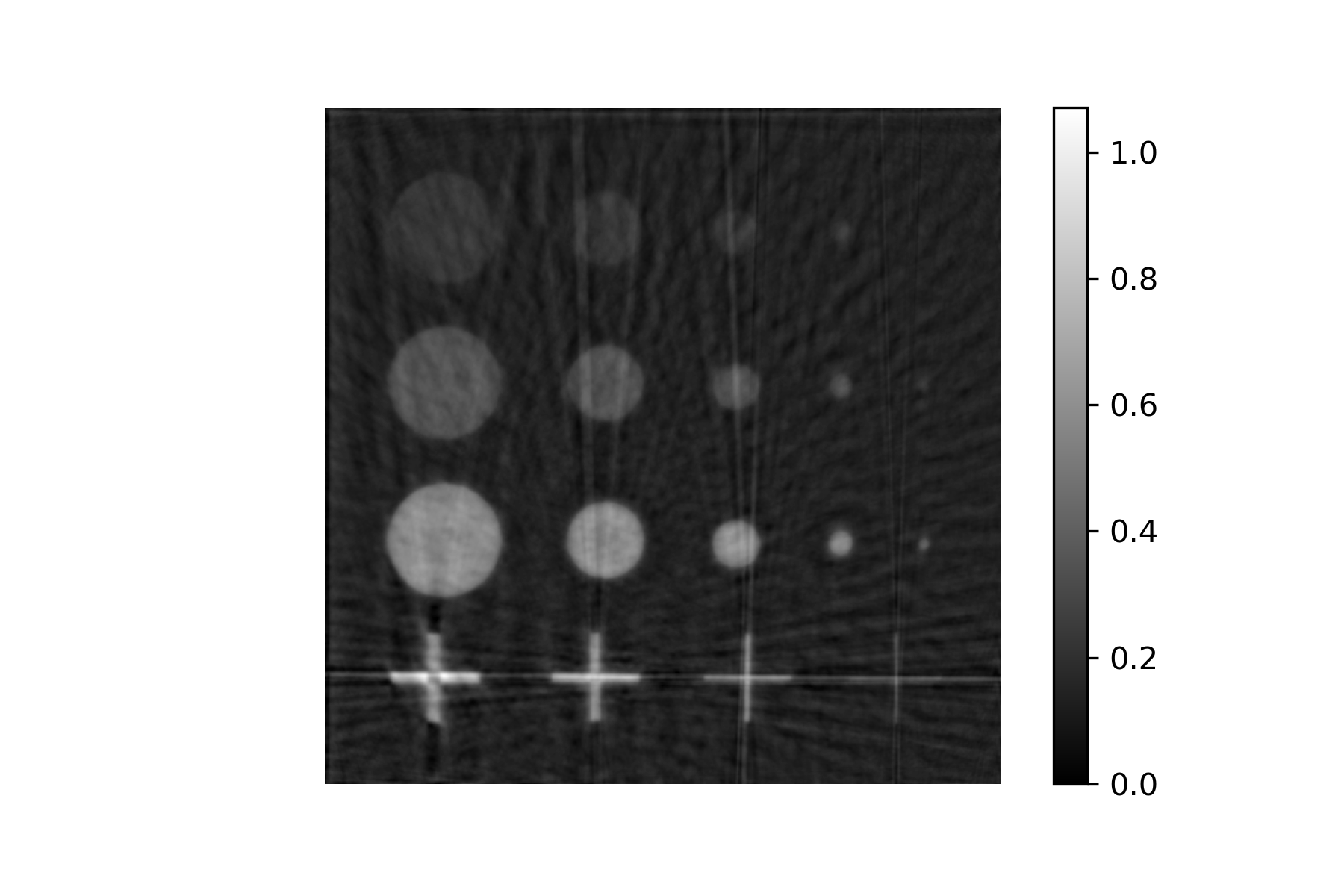}\\
    \includegraphics[trim={35mm 0 15mm 0},clip, width=0.3\linewidth]{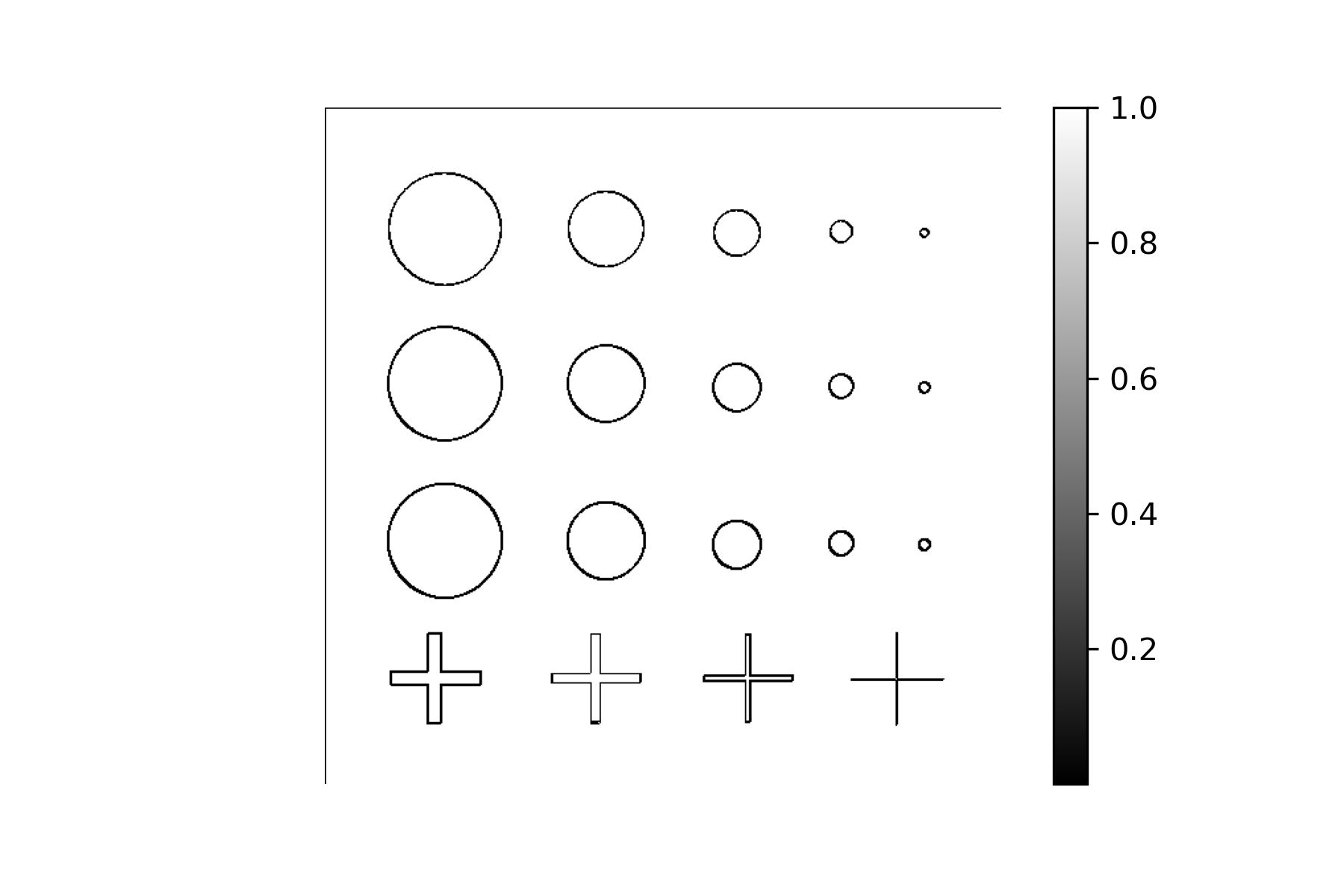}
    \includegraphics[trim={35mm 0 15mm 0},clip, width=0.3\linewidth]{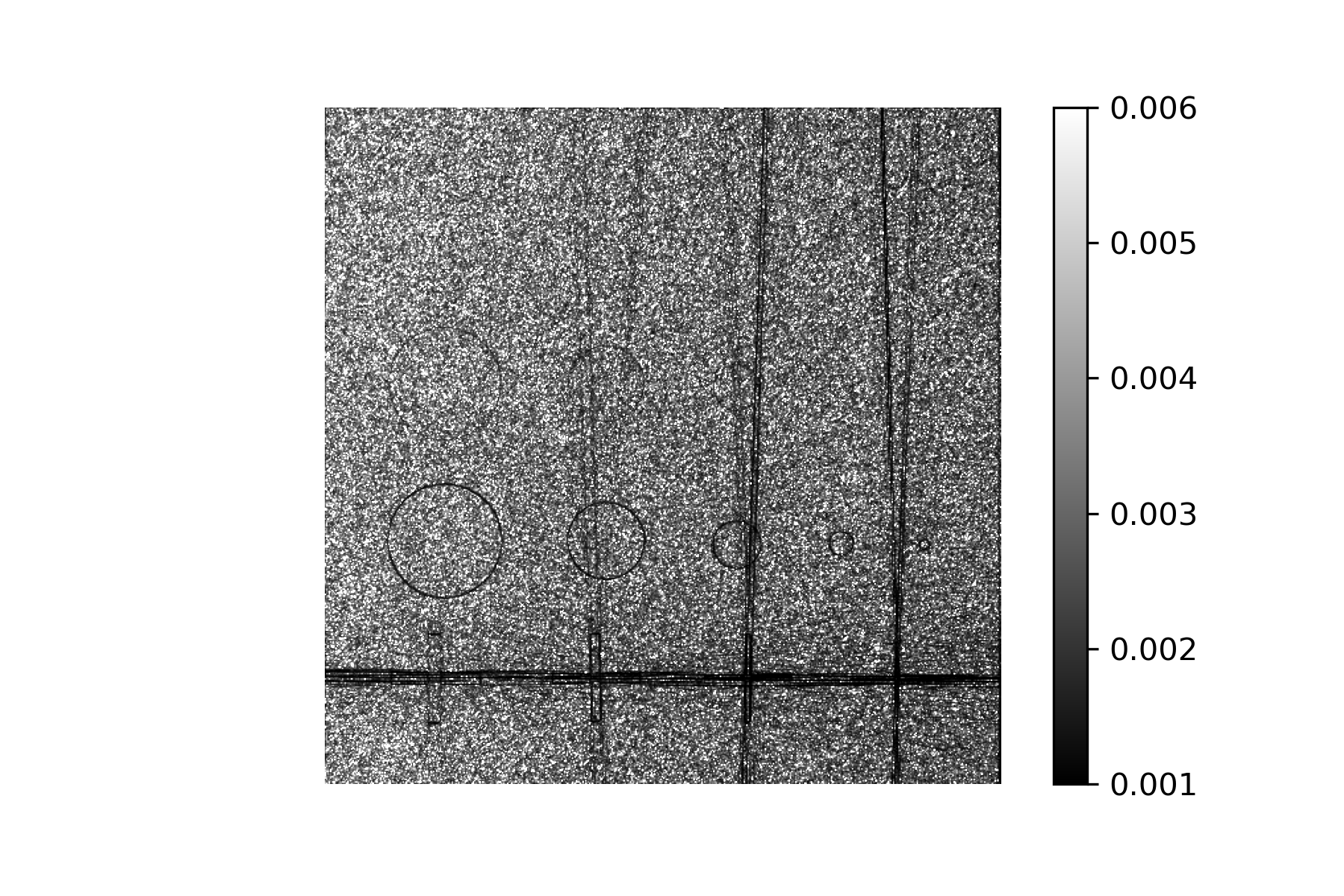}
    \includegraphics[trim={35mm 0 15mm 0},clip, width=0.3\linewidth]{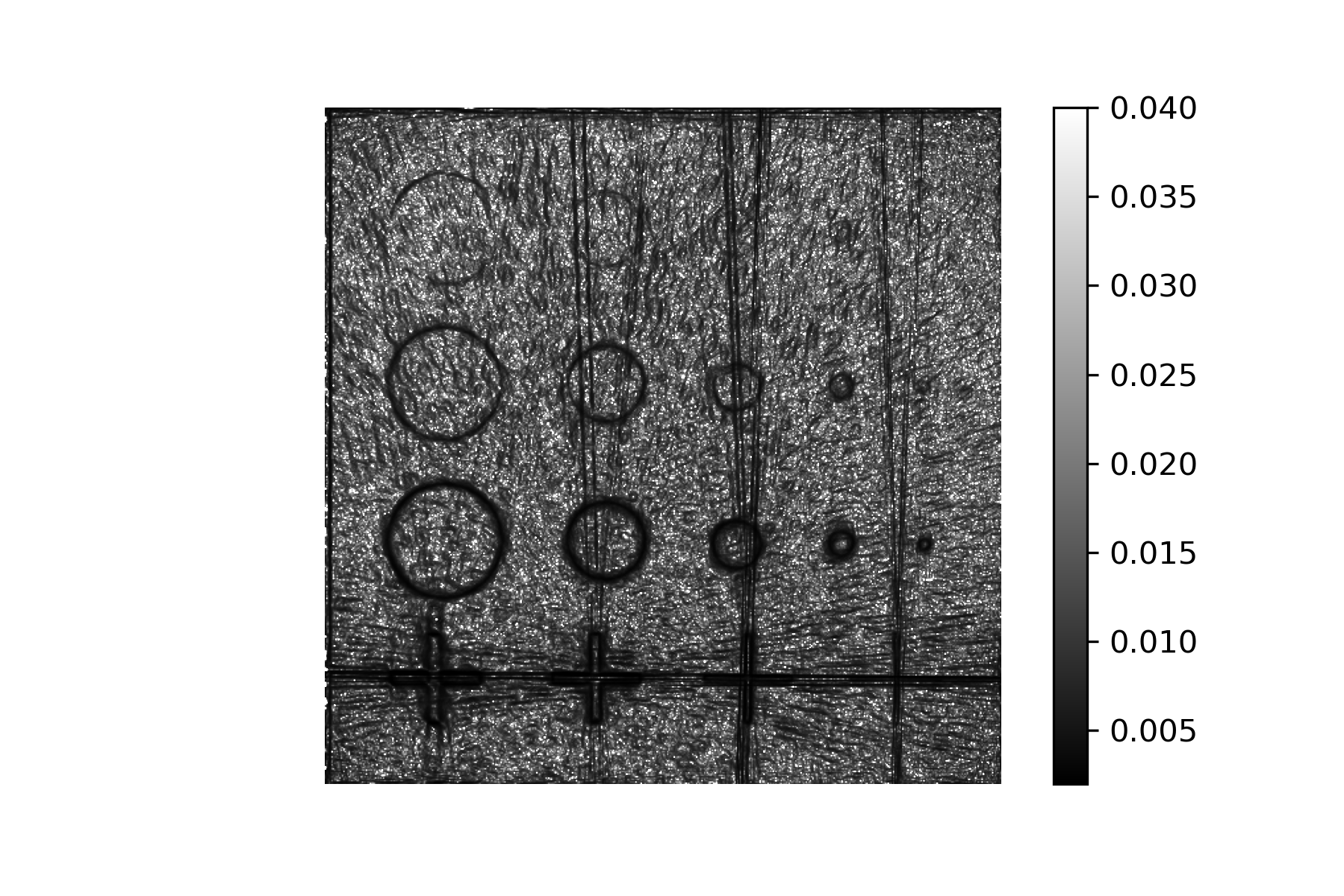} \\ 
  \caption{Experiment on the synthetic image with low noise ($\nu=0.005$). In the first row, from left to right: $\tilde{\x}^{GT}$, $\tilde{\x}^{FBP}$, $\tilde{\x}^{TV}$. In the second row: the corresponding images of weights $\w_\eta(\tilde{\x})$.}
  \label{fig:results_Sint005}
\end{figure}

In Figure \ref{fig:results_Sint005} we display the $\tilde{\x}$ images in the upper row and the corresponding weights $\w_\eta(\tilde{\x})$ in the lower row. The images computed by the FBP and TV reconstructors exhibit small noise but severe artifacts represented by streaks, which are due to the low number of views (45 out of 180 degrees). These artifacts are inherited by the $\w_\eta(\tilde{\x})$ weight images. The colorbar provides information about the range of the reconstructed values by the various methods. However, a quite large value of $\lambda=5$ inhibits these defects and leads to a visually nearly perfect reconstruction (SSIM $> 0.99$). \\
In the subsequent test, we increased the noise ($\nu=0.02$)  and we set $\lambda=10$ and $\eta=2 \cdot 10^{-5}$ for all the methods except the case where $\Psi$ is the FBP algorithm where $\eta=2 \cdot 10^{-3}$.
The results are reported in the lower part of Table \ref{tab:Sintetica} and the images are displayed in in Figure \ref{fig:results_Sint02}. 
The achieved results exhibit overall excellence in terms of SSIM (above all considering the sparsity of the tomographic data). However, as depicted in Figure \ref{fig:results_Sint02}, which showcases a crop of the restored images, discernible differences are evident in output, particularly in the reconstruction of the finer cross details and contrast. Notably, the global TV method consistently fails to reconstruct that cross regardless of the parameter value. The superior performance of the reconstruction when $\tilde{\x}={\x}^{GT}$ underscores the efficacy of adapting regularization to image pixels when provided with a highly accurate approximation of the target image.\\
Figure \ref{fig:results_Sint02} also includes plots illustrating the behavior of the objective function (on the left) and Relative Error (on the right) across iterations for the various examined approaches. It is noteworthy that, despite comparing distinct objective functions in the first plot, their trends exhibit similar descending patterns, ultimately converging to nearly identical values.
In regards to the error plot, the order of the curves is maintained across all iterations, consistent with the sequence discussed in Table \ref{tab:Mayo005}. The plots clearly indicate that the error ceases to decrease further, indicating that the best achievable reconstruction has been obtained for each method.

\begin{figure}
\centering
\includegraphics[trim={65mm 0 0 65mm},clip, width=0.23\linewidth]{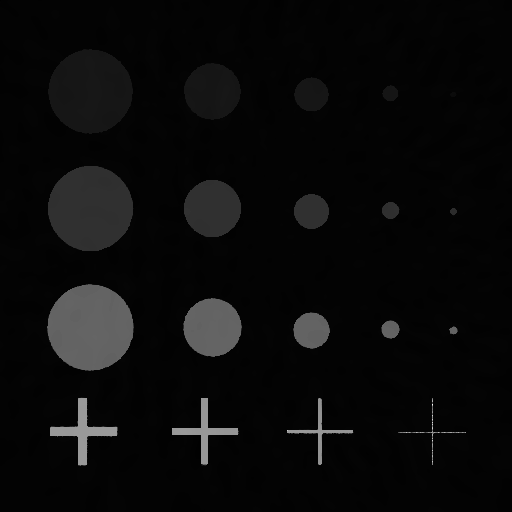}  \ 
\includegraphics[trim={65mm 0 0 65mm},clip, width=0.23\linewidth]{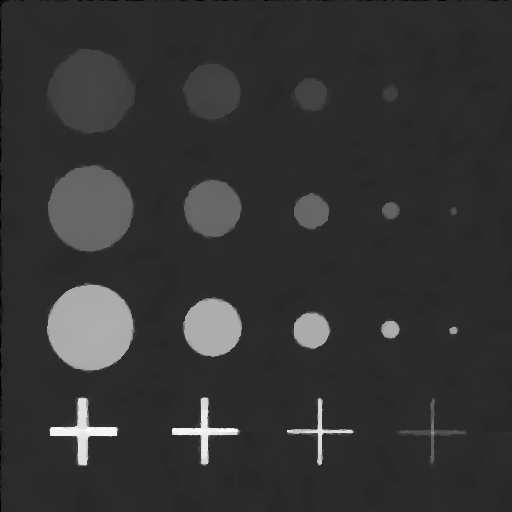}  \
\includegraphics[trim={65mm 0 0 65mm},clip, width=0.23\linewidth]{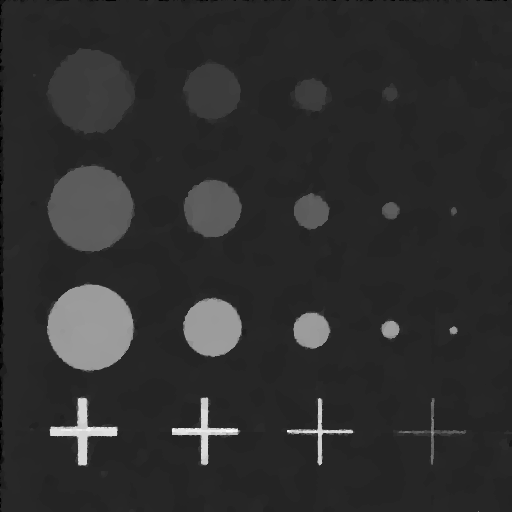} \ 
\includegraphics[trim={65mm 0 0 65mm},clip, width=0.23\linewidth]{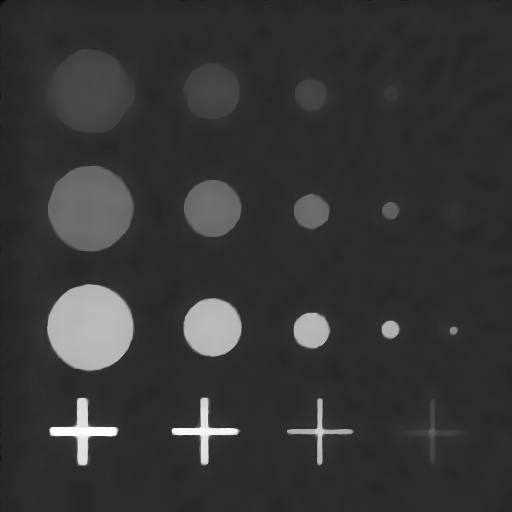} \\
\includegraphics[width=0.45\linewidth]{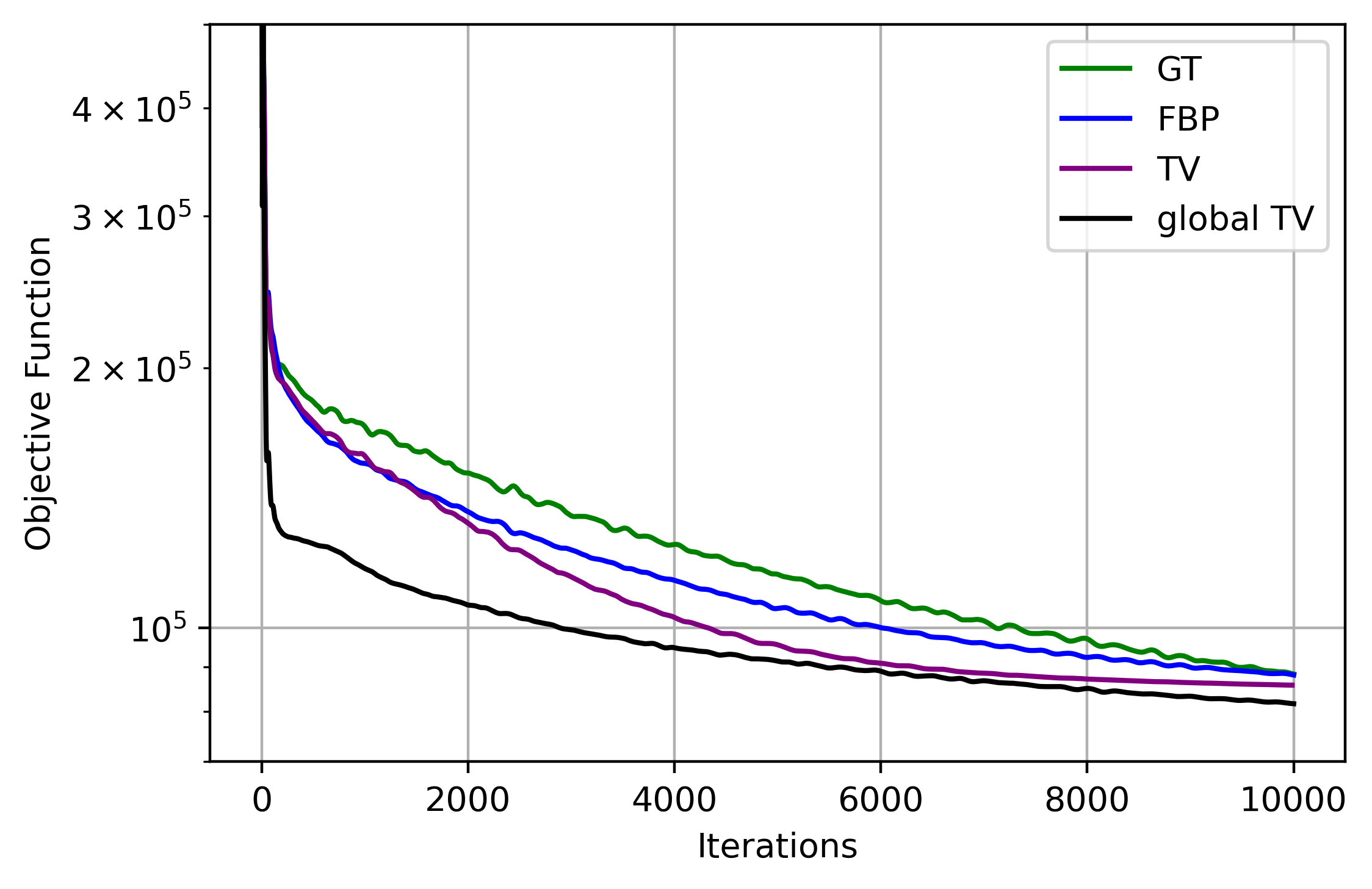}
\includegraphics[width=0.44\linewidth]{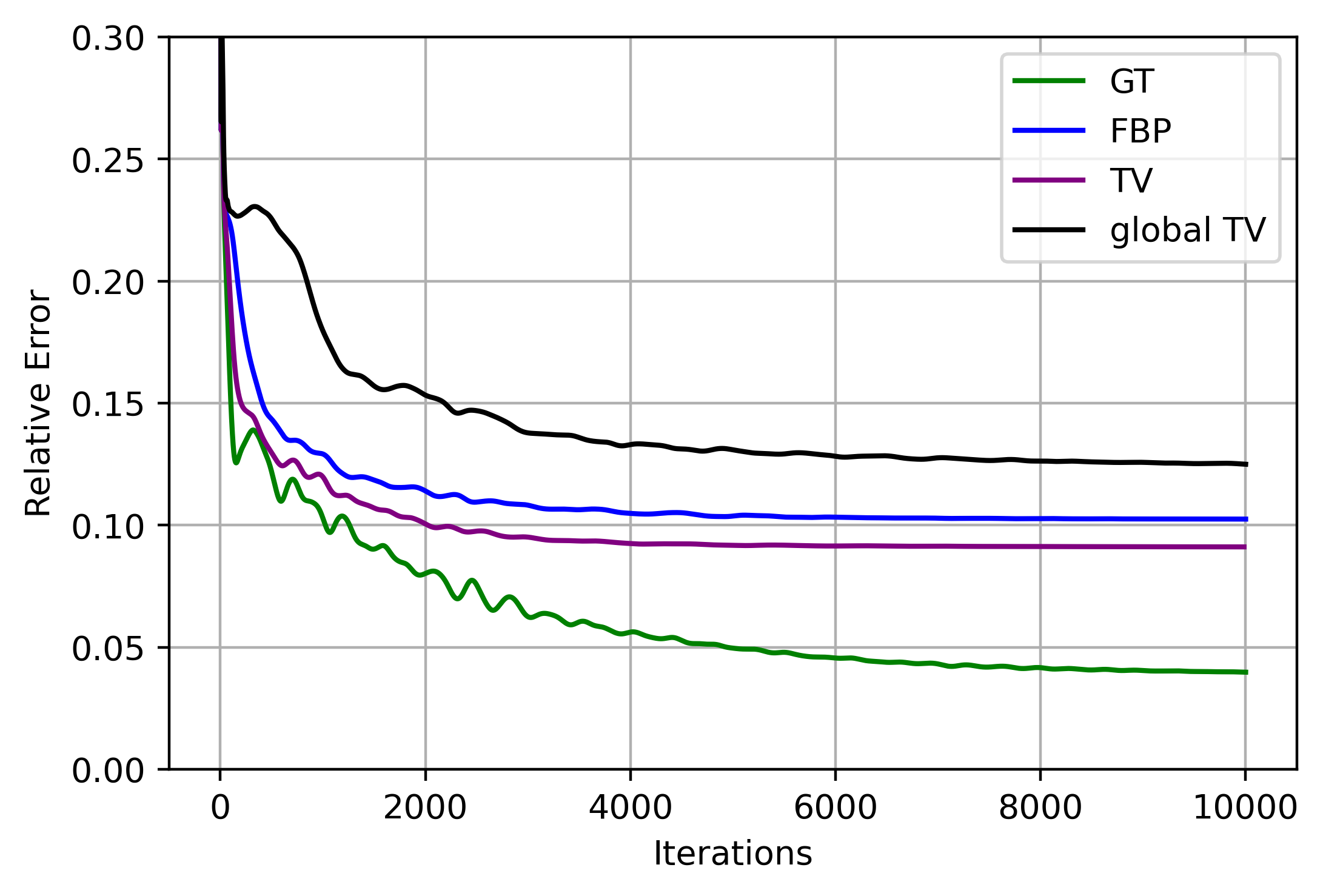}
  \caption{Results of the experiment on the synthetic image with high noise ($\nu=0.02$). In the first row, from left to right: zooms on the lower-right angle of the reconstructions    GT-$W\ell_1$, FBP-$W\ell_1$, TV-$W\ell_1$, global TV. In the second row: plot of the objective function (left) and of the Relative Error  (right) over the iterations.}
  \label{fig:results_Sint02}
\end{figure}


\subsection{Experiments on real medical images}\label{ssec:Mayo}

The second test image is a real chest tomographic image  from Low Dose CT Grand Challenge data set by the Mayo Clinic, which is not part of the training set used for our neural network. The image is  characterized by low-contrast regions and by high-contrast tiny details in the lungs.  It is depicted  with  its  crop  in Figure \ref{fig:MayoGT}. The aim of these experiments is firstly to test the reconstructor $\Psi_{Theta}$ constituted by the NN, and secondly to analyze the methods' performance on a real image with a gradient less sparse than in the synthetic, previous case.\\
We created the test problem as in the synthetic simulation and initially fixed $\nu=0.005$. 
From the metrics reported in Table \ref{tab:Mayo005} we observe that the NN-$W\ell_1$ approach always overcomes the FBP-$W\ell_1$ in final quality. Regarding the elastic-loss of the neural network defined in \eqref{eq:ElasticLoss}, the best outcomes are attained with $\alpha=1$, wherein the network learns merely gradients instead of the image. In fact, the $\tilde{\x}$ image is of low quality, yet the computed weights derived from it yield highly accurate reconstructions. Interestingly, setting $\alpha=0.5$ gives the best $\tilde\x$ and the corresponding final solution has metrics very close to one achieved with $\alpha=1$. In addition, in these two cases the behavior of the RE metric is highly comparable over the iterations, as visible in the plot in Figure \ref{fig:MayoGT}. \\

In the final experiment, we introduce high noise by setting $\nu=0.02$ and we apply the $\Psi$-W$\ell_1$ methods.  In this case we start from the same image provided by the $\Psi_{\Theta}$ reconstructor (with $\alpha=0.5$) and we consider various values for the parameter  $\eta=2 \cdot 10^{-3},2 \cdot 10 ^{-4}, 2 \cdot 10^{-5}$. Table \ref{tab:Mayo02} demonstrates that the value of $\eta$ exerts a discrete influence on the quality of the output image, as also depicted in Figure \ref{fig:results_Mayo02_w}. The initial column of the figure displays the masks in descending order of $\eta$ from top to bottom. 
The second and third columns of Figure \ref{fig:results_Mayo02_w} correspond to the final reconstruction. The top image emphasizes the edges of the objects, the middle one smoothens the background to a value close to one, and the last one produces an excessively smooth image. Thus, the optimal reconstruction is attained with $\eta=2 \cdot 10^{-4}$, representing the optimal compromise among the aforementioned reconstructions

\begin{figure}
\centering
\includegraphics[width=0.24\linewidth]{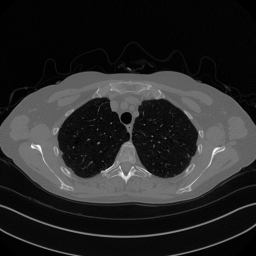}
\includegraphics[trim={35mm 10mm 10mm 35mm},clip, width=0.24\linewidth]{imm_Mayo_005/mayoGT.png}
\includegraphics[width=0.4\linewidth]{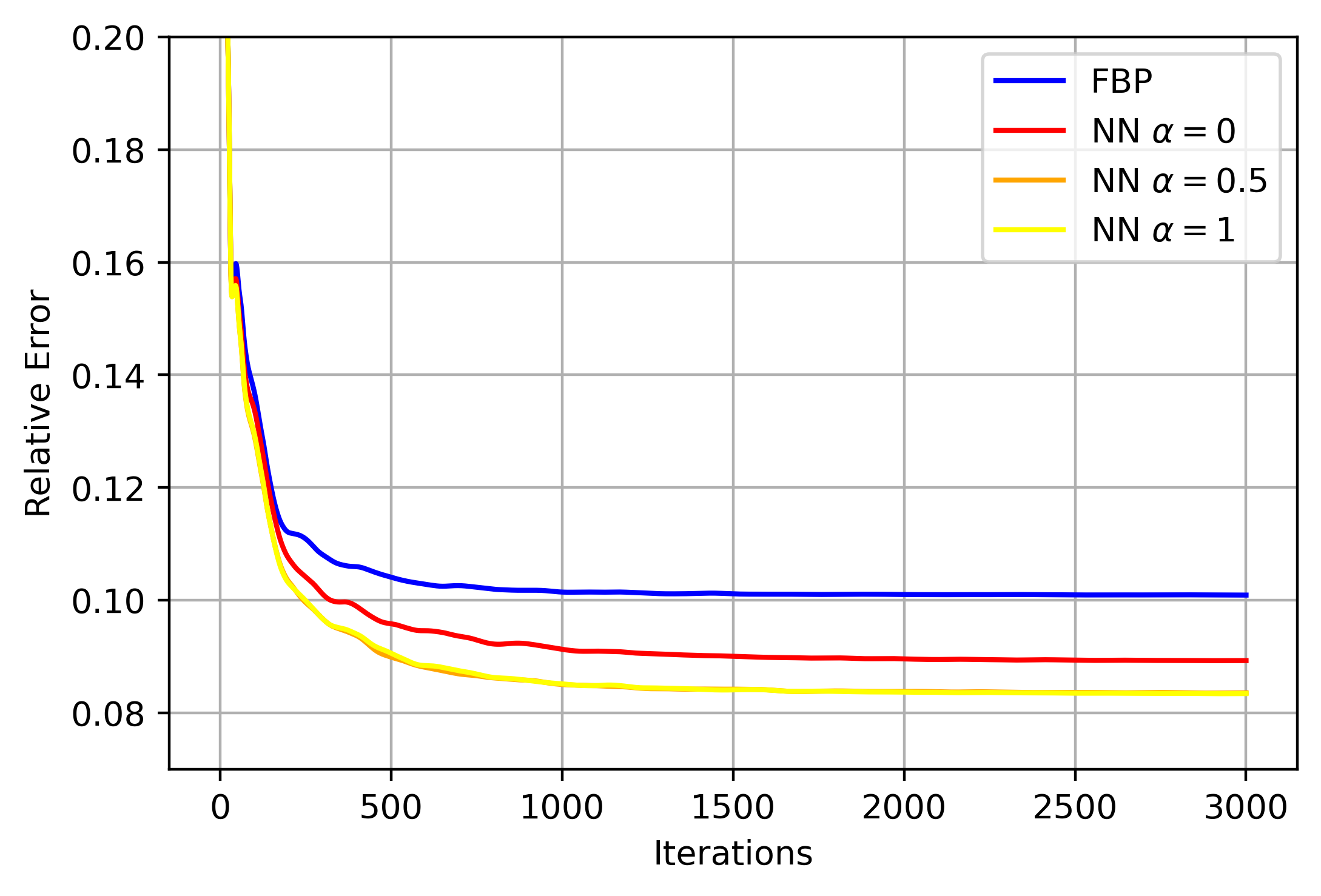}
  \caption{From left to right: ground truth image from Mayo data set, a cropped portion of it, a plot of the RE over iterations for the reconstructions  by FBP-$W\ell_1$ and NN-$W\ell_1$ with $\alpha=0, 0.5, 1$ ($\nu=0.005$).}
  \label{fig:MayoGT}
\end{figure}



\begin{table}[h]
\caption{Performance results on the real image with low ($\nu=0.005$) noise. In the first three columns the metrics relative to the image $\tilde{\x}$, in the last three columns the metrics relative to the output image $\x^*_{\Psi,\delta}$. }\label{tab:Mayo005}%
\begin{tabular*}{\textwidth}{@{\extracolsep\fill}l rrr rrr}
\toprule
 & \multicolumn{3}{@{}c@{}}{ $\tilde{\x} $} & \multicolumn{3}{@{}c@{}}{$\x^*_{\Psi,\delta}$}\\
 & RE & PSNR   & SSIM & RE & PSNR   & SSIM\\
\cmidrule{2-4} \cmidrule{5-7}
\midrule
FBP-$W\ell_1$                     & 0.2879 & 23.1629 & 0.4133         & 0.1009 & 32.2698 & 0.8832  \\
NN-$W\ell_1$  ($\alpha=0$)    & 0.1038 & 32.0253 & 0.8588         & 0.0893 & 33.3329 & 0.9009  \\ 
NN-$W\ell_1$  ($\alpha=0.5$)  & 0.0854 & 33.7229 & 0.8994         & 0.0836 & 33.9087 & 0.9121  \\
NN-$W\ell_1$  ($\alpha=1$)    & 0.5092 & 18.2113 & 0.3843         & 0.0834 & 33.9224 & 0.9128  \\    
\bottomrule
\end{tabular*}
\end{table}



\begin{table}[h]
\caption{Performance results on the real image with high  noise ($\nu=0.02$). In the first three columns the metrics relative to the image $\tilde{\x}$, in the last three columns the metrics relative to the output image $\x^*_{\Psi,\delta}$. }\label{tab:Mayo02}%
\begin{tabular*}{\textwidth}{@{\extracolsep\fill}l rrr rrr}
\toprule
 & \multicolumn{3}{@{}c@{}}{ $\tilde{\x} $} & \multicolumn{3}{@{}c@{}}{$\x^*_{\Psi,\delta}$}\\
            & RE & PSNR   & SSIM               & RE & PSNR   & SSIM\\
\cmidrule{2-4} \cmidrule{5-7}
\midrule
FBP-$W\ell_1$                       & 0.4910 & 18.5280 & 0.2118         & 0.1431 & 29.2351 & 0.8064  \\
NN-$W\ell_1$ ($\eta=2 \cdot 10^{-4}$)   & 0.1052 & 31.9065 & 0.8513         & 0.1473 & 28.9847 & 0.8078  \\ 
NN-$W\ell_1$ ($\eta=2 \cdot 10^{-3}$)   & 0.1052 & 31.9065 & 0.8513         & 0.1214 & 30.6619 & 0.8546  \\
NN-$W\ell_1$ ($\eta=2 \cdot 10^{-2}$)   & 0.1052 & 31.9065 & 0.8513         & 0.1219 & 30.6318 & 0.8421  \\    
\bottomrule
\end{tabular*}
\end{table}

\begin{figure}
\centering

\includegraphics[trim={35mm 11mm 15mm 0},clip, width=0.33\linewidth]{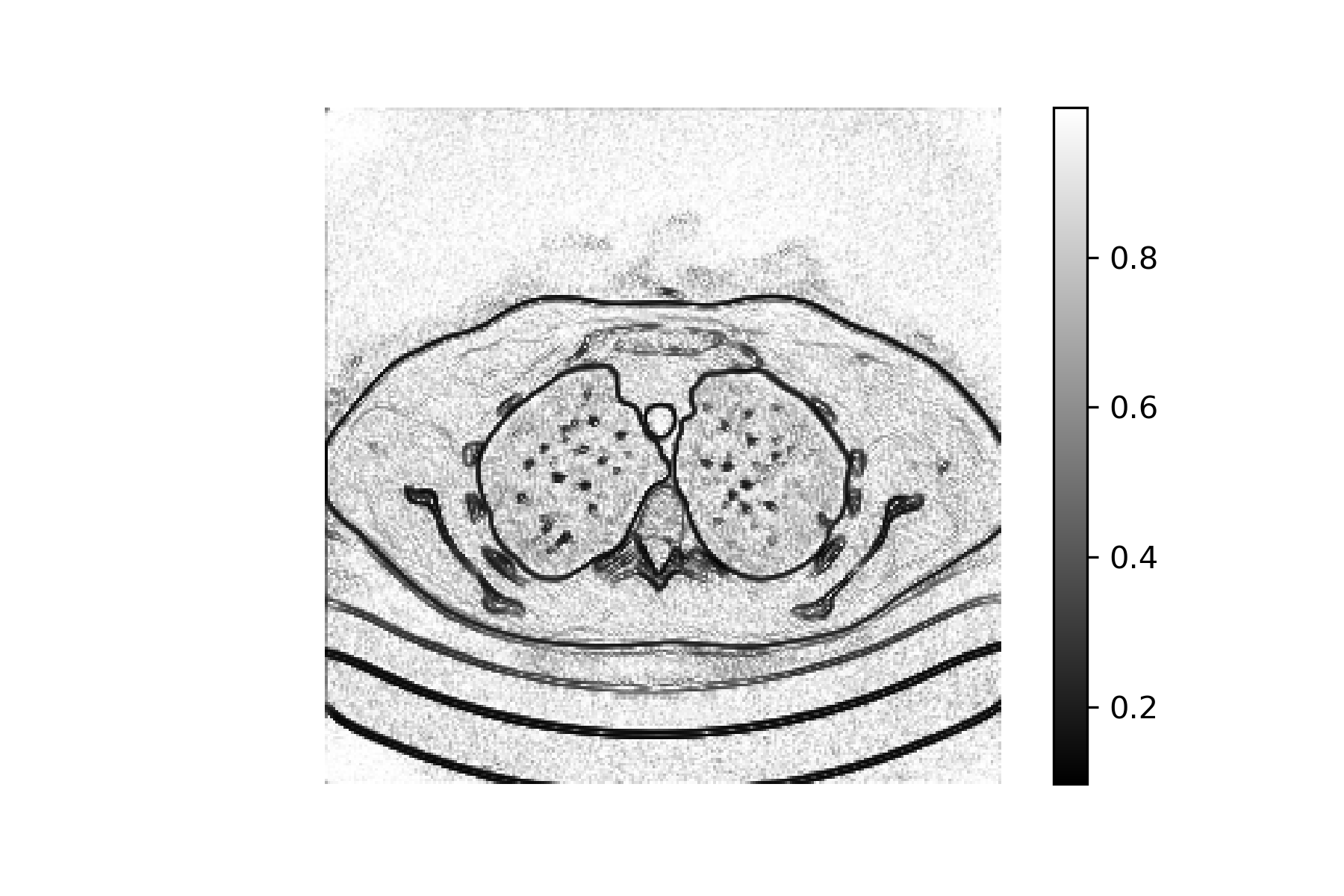}
\includegraphics[trim={35mm 11mm 15mm 0},clip, width=0.33\linewidth]{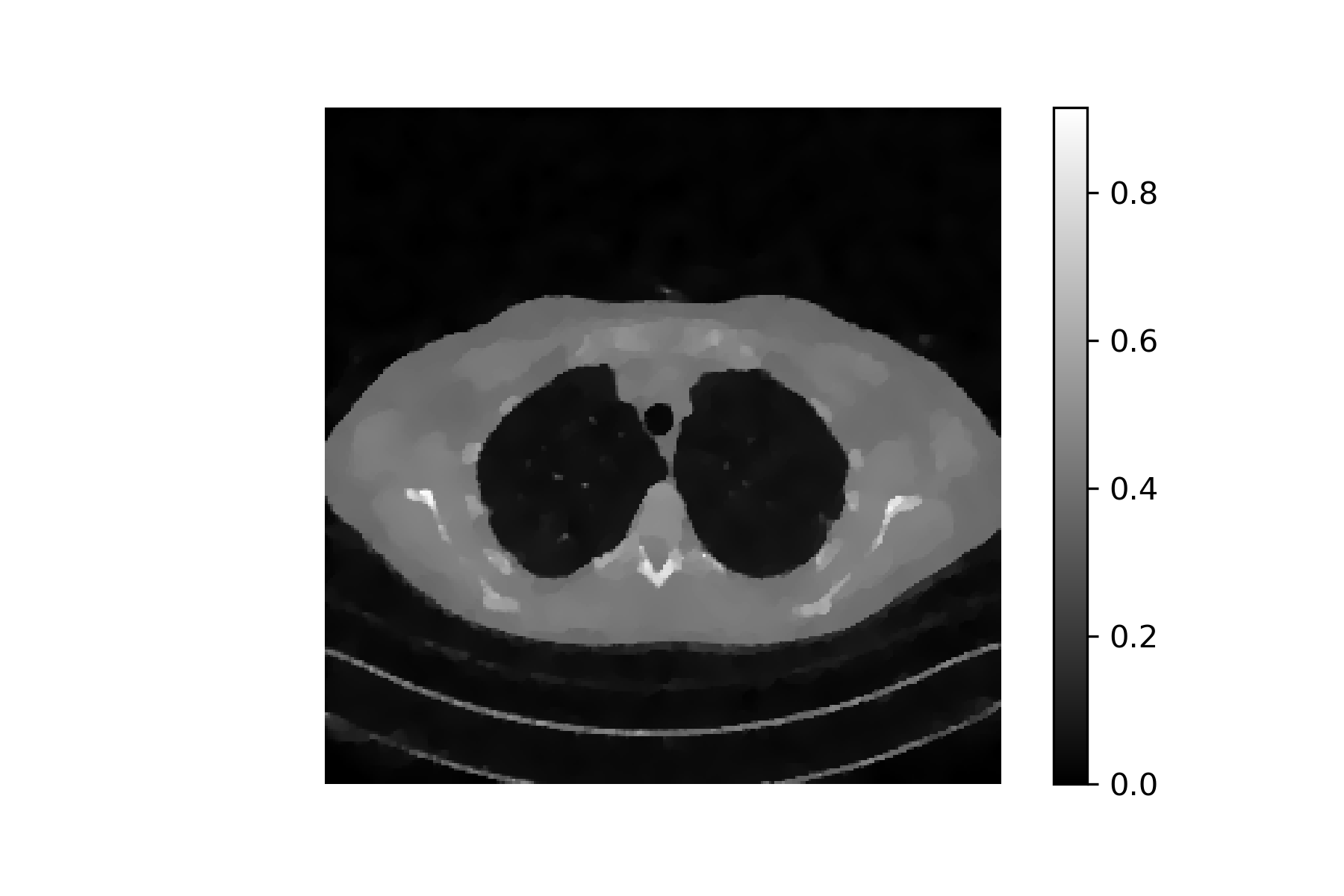}
\includegraphics[trim={35mm 10mm 10mm 35mm},clip, width=0.25\linewidth]{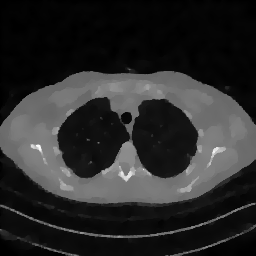} \\
\includegraphics[trim={35mm 11mm 15mm 0},clip, width=0.33\linewidth]{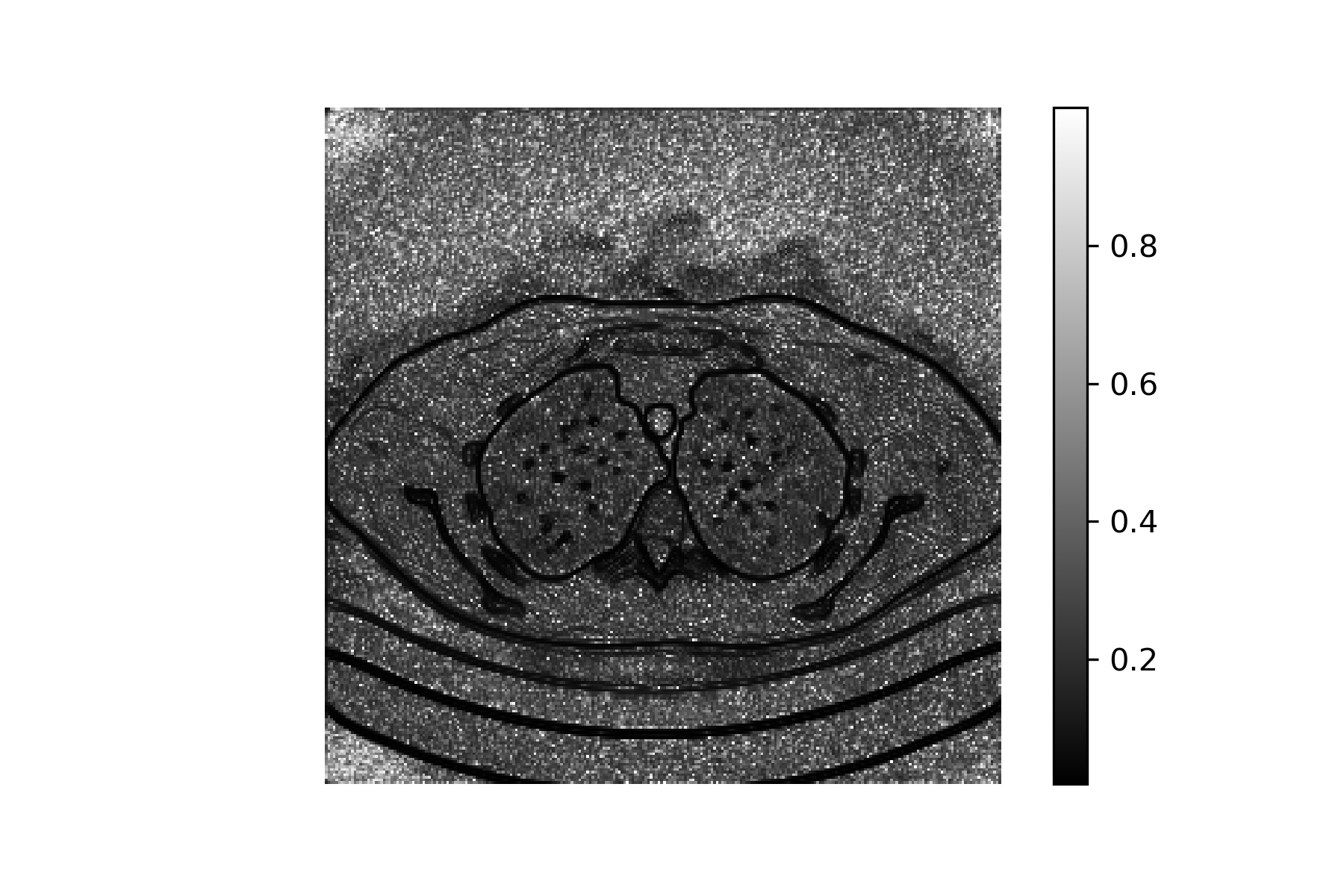}
\includegraphics[trim={35mm 11mm 15mm 0},clip, width=0.33\linewidth]{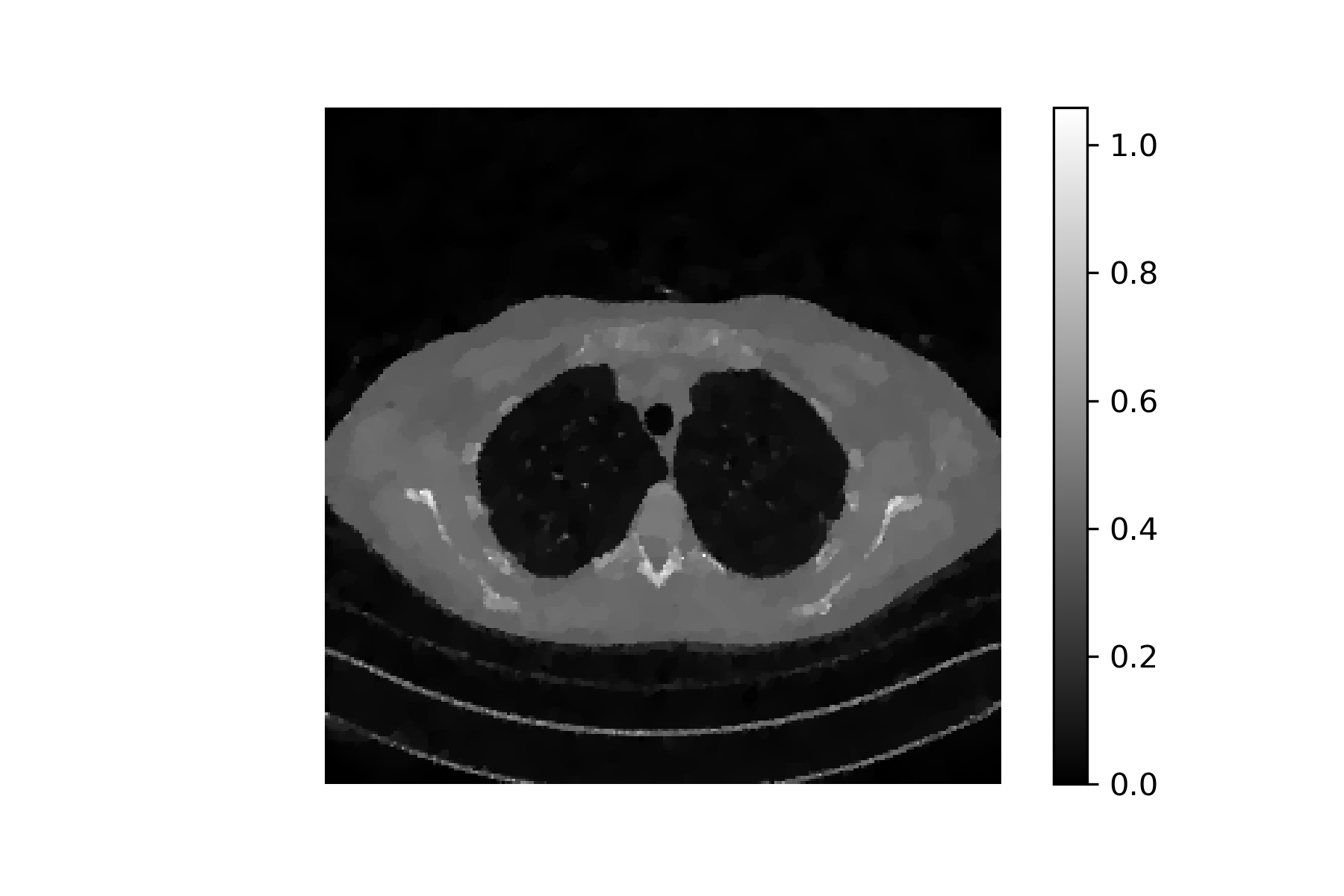}
\includegraphics[trim={35mm 10mm 10mm 35mm},clip, width=0.25\linewidth]{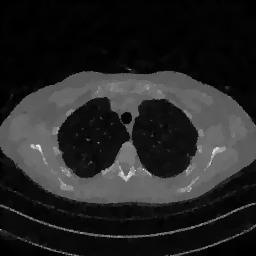} \\
\includegraphics[trim={35mm 11mm 15mm 0},clip, width=0.33\linewidth]{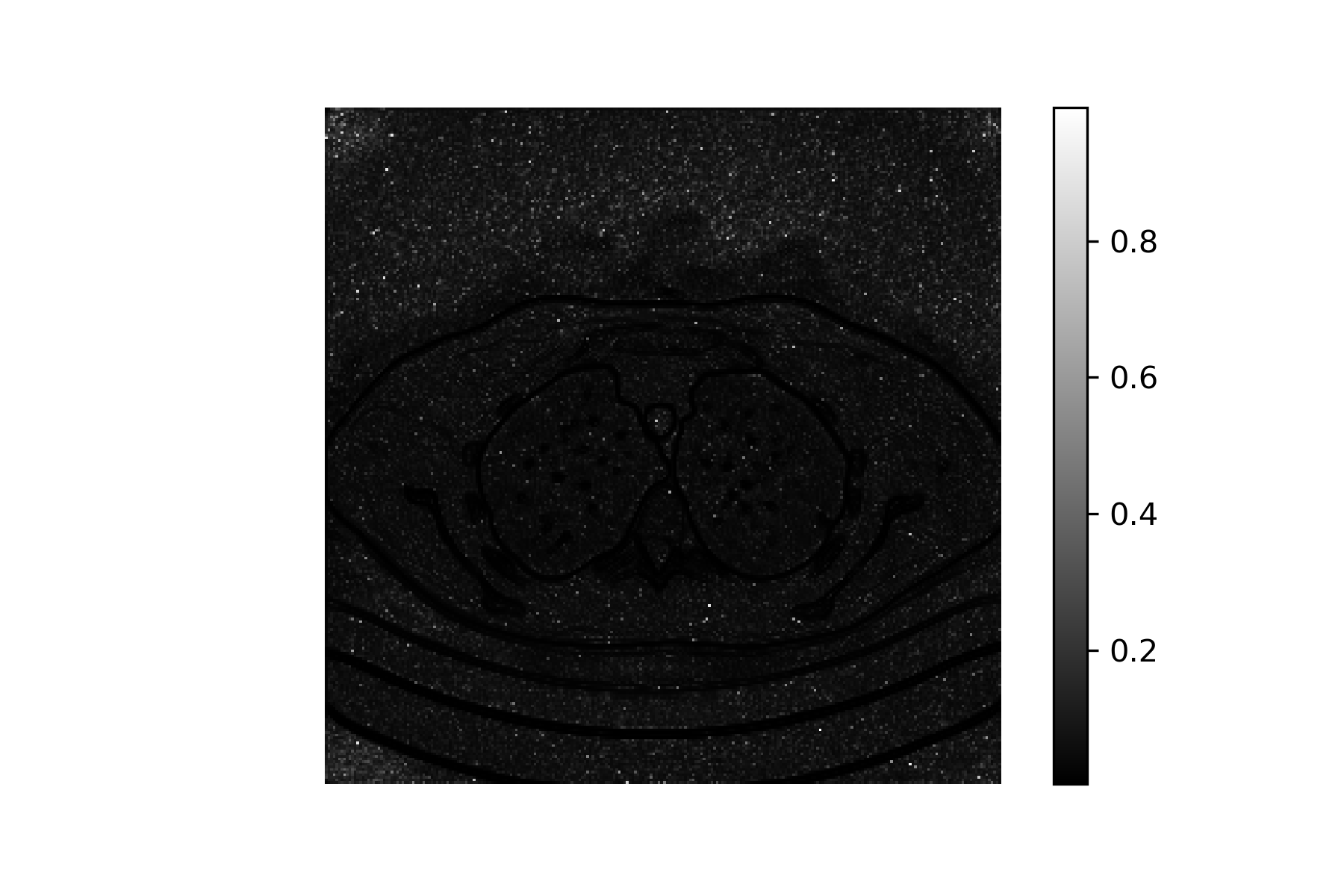}
\includegraphics[trim={35mm 11mm 15mm 0},clip, width=0.33\linewidth]{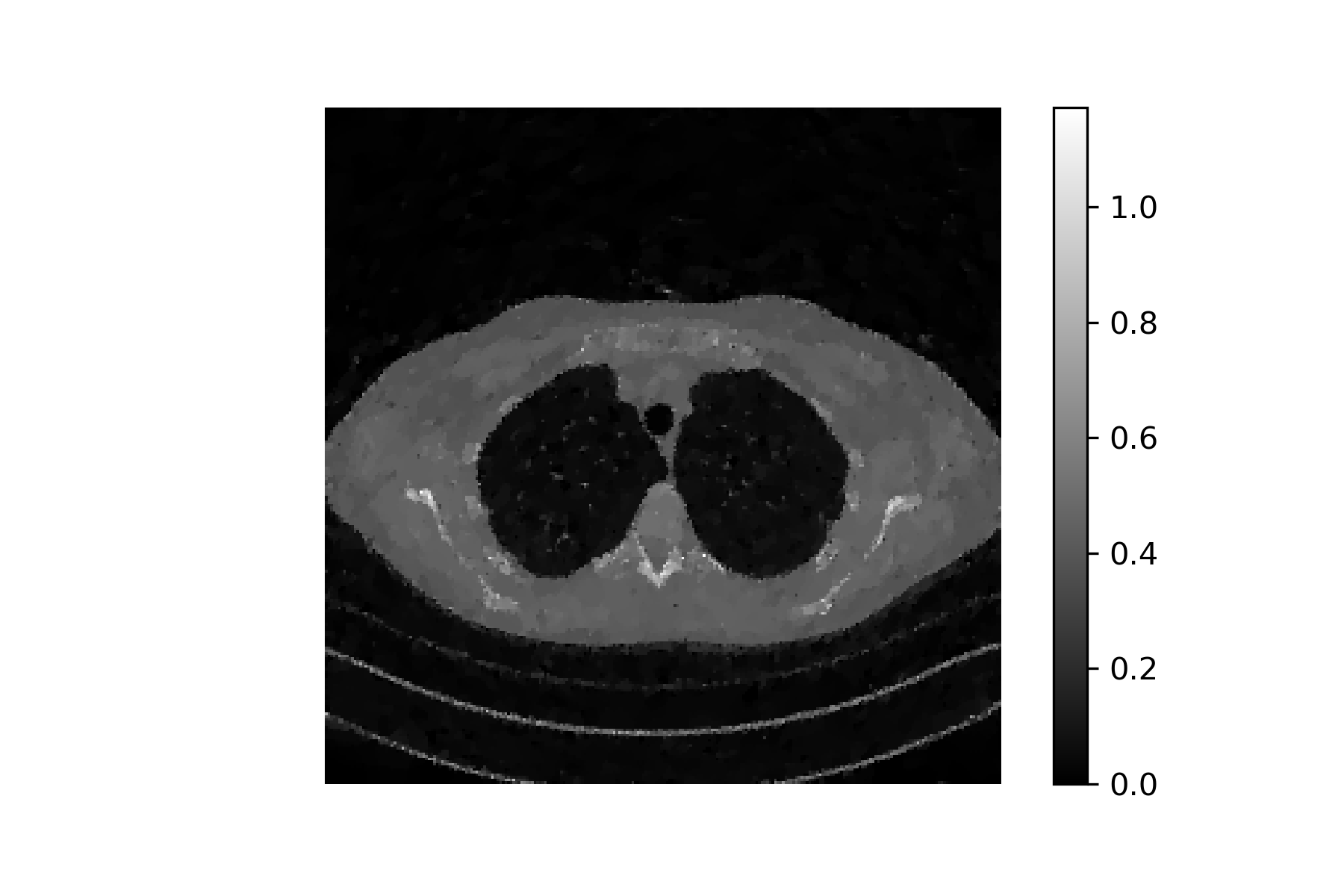} 
\includegraphics[trim={35mm 10mm 10mm 35mm},clip, width=0.25\linewidth]{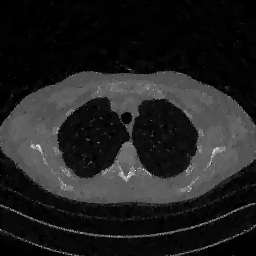}
  \caption{Results of the experiment on the real image with high noise ($\nu=0.02$) with NN-$W\ell_1$. In the left column from top to down the images of the weights with $\eta=2 \cdot 10^{-2}$, $\eta=2 \cdot 10^{-3}$, $\eta=2 \cdot 10^{-4}$; in the central column the corresponding reconstructions and in the right column a cropped portion of them.}
  \label{fig:results_Mayo02_w}
\end{figure}

\section{Conclusion}\label{sec:Concl}

In this paper, we propose a space-variant weighted Total Variation regularization method for reconstructing tomographic images from sparse views. Sparse-view CT constitutes an under-determined linear inverse problem, the solution of which is marked by noise and prominent artifacts.
The proposed weighting strategies is obtained by linearizing a smoothed Total p-Variation, and by exploiting an image that we called $\Tilde x$, without any prior knowledge of the noise characteristics in the data.
We have used a data driven approach to effectively approximate the ground truth image or its gradients. Finally we have applied the space-variant Total Variation to a sinogram from  very sparse geometry both of a synthetic and real image. The numerical results show that our model outperforms the global Total Variation, which uniformly acts on all the pixels of the image. Moreover, the best result is obtained when the information of the weights is retrieved from the ground truth, as demonstrated in the theoretical analysis. The approximation achieved by the neural network demonstrates its effectiveness, particularly when it learns to reconstruct the gradients of the images.
The proposed spatially variant weighted approach can be easily extended to regularization methods other than Total Variation. In addition, we also intend to adaptively select the value of $p$ in \eqref{eq:Psi_WL1_formulation}, in order to appropriately regularize both smooth and non-smooth regions of the image.

\section*{Declarations}

\subsection{Fundings} 
This work was partially supported by “Gruppo Nazionale per il Calcolo Scientifico (GNCS-INdAM)” (Progetto 2023 “Modelli e metodi avanzati in Computer Vision"). \\
E. Loli Piccolomini, D. Evangelista and E. Morotti are supported by the PRIN 2022 project “STILE: Sustainable Tomographic Imaging with Learning and rEgularization”, project  code: 20225STXSB, funded by the European Commission under the NextGeneration EU programme.\\
A. Sebastiani is supported by the project “PNRR - Missione 4 “Istruzione e Ricerca” - Componente C2 Investimento 1.1 “Fondo per il Programma Nazionale di Ricerca e Progetti di Rilevante Interesse Nazionale (PRIN)”, “Advanced optimization METhods for automated central veIn Sign detection in multiple sclerosis from magneTic resonAnce imaging (AMETISTA)”, project code: P2022J9SNP, MUR D.D. financing decree n. 1379 of 1st September 2023 (CUP E53D23017980001) funded by the European Commission under the NextGeneration EU programme.

\paragraph{Conflict of interest.} The authors declare that they have no conflict of interest in relation to the research presented in this paper.

\paragraph{Data availability.} Under request to the authors.

\paragraph{Materials availability.} Not applicable.

\paragraph{Code availability.} Under request to the authors.

\begin{appendices}

\section{Theoretical derivation of convergence results}\label{sec:Appendix}

In the following, we prove a stability result showing that $|| \x^*_{GT, \delta} - \x^*_{\Psi, \delta} ||_1$, i.e. the error introduced by using $\Psi(\y^\delta)$ instead of $\x^{GT}$ into the computation of $\w_\eta(\tilde{\x})$, converges to $0$ as the distance $|| \Psi(\y^\delta) - \x^{GT} ||_1$ goes to 0.

We first need to introduce a preliminary result.
\begin{proposition}\label{prop:uniformconvergence_implies_convergenceminimizers}
    Let $\{ f_k \}_{k \in \mathbb{N}}$ be a sequence of functions converging uniformly to $f^*$ on any compact subset of $\X$. Let us assume that for any $k \in \mathbb{N}$ it holds:
    \begin{enumerate}[label=(\roman*)]
        \item $\min_{\x \in \X} f_k(\x)$ has a unique solution $\x_k \in \R^n$, 
        \item $\min_{\x \in \X} f^*(\x)$ has a unique solution $\x^* \in \R^n$.
    \end{enumerate}
    Then:
    \begin{align*}
        \lim_{k \to \infty} || \x_k - \x^* ||_1 = 0.
    \end{align*}
\end{proposition}

\begin{proof}
    Let $\x_k$ be the unique minimizer of $f_k$ for any $k \in \mathbb{N}$ and let $C \subseteq \X$ be any compact set such that $\x_k \in C$ for any $k \in \mathbb{N}$. Since $\{ \x_k \}_{k \in \mathbb{N}}$ is a sequence in a compact set, there exists a convergent subsequence $\{ \x_{k_j} \}_{j \in \mathbb{N}}$. Let $\x^*$ be the limit of this subsequence. By definition, for any $k \in \mathbb{N}$, it holds:
    \begin{align*}
        f_k(\x_k) \leq f_k(\x), \quad \forall \x \in C,
    \end{align*}
    and in particular:
    \begin{align*}
        f_{k_j}(\x_{k_j}) \leq f_{k_j}(\x), \quad \forall \x \in C.
    \end{align*}
    Since $f_k$ converges uniformly to $f^*$ and $\x_{k_j}$ converges to $\x^*$, it holds that:
    \begin{align*}
        &\lim_{j \to \infty} f_{k_j}(\x_{k_j}) = f^*(\x^*), \\
        &\lim_{j \to \infty} f_{k_j}(\x) = f^*(\x), \quad \forall \x \in C.
    \end{align*}
    Consequently,
    \begin{align*}
        f^*(\x^*) \leq f^*(\x), \quad \forall \x \in C,
    \end{align*}
    i.e. $\x^*$ is a minimizer of $f^*$. Since $f^*$ has a unique minimizer, then all the subsequences of $\x_k$ converge to the same element $\x^*$, which implies that $\x_k$ converges to $\x^*$ as $k \to \infty$. In particular:
    \begin{align*}
        \lim_{k \to \infty} || \x_k - \x^* ||_1 = 0.
    \end{align*}
\end{proof}

Now, revisiting the definition of the $\Psi$ reconstructor introduced in Section \ref{ssec:conv}, we expound upon the notions of $p$-norm accuracy $\eta_p^{-1}$ and $p$-norm $\epsilon$-stability constant $C^\epsilon_{\Psi, p}$ associated with $\Psi$, thereby extending the concepts initially introduced in \cite{evangelista2022or,evangelista2023ambiguity}.
%
\begin{definition}\label{def:accuracy}
    Let $\Psi: \R^m \to \R^n$ be a reconstructor. Let:
    \begin{align*}
        \eta_p := \sup_{\x \in \X} || \Psi(\K \x) - \x ||_p.
    \end{align*}
    If $\eta_p < \infty$, we say the $\Psi$ is $\eta_p^{-1}$-accurate in $p$-norm.
\end{definition}

\begin{definition}\label{def:stability_constant}
    Let $\Psi: \R^m \to \R^n$ be an $\eta_p^{-1}$-accurate reconstructor in $p$-norm. Then, for any $\epsilon > 0$, we say that $C^\epsilon_{\Psi, p}$ is the $p$-norm $\epsilon$-stability constant if:
    \begin{align*}
        C^\epsilon_{\Psi, p} = \sup_{\substack{\x \in \X \\ || \e ||_p \leq \epsilon}} \frac{|| \Psi(\K\x + \e) - \x ||_p - \eta_p}{|| \e ||_p}.
    \end{align*}
\end{definition}
For an analysis of the accuracy and the stability constant of a reconstructor, see \cite{evangelista2022or} where the case $p=2$ is explored in detail. Note that most of the results obtained for $p=2$ can be easily extended to $p\geq 1$. 
In the context of this study, we emphasize that $\eta_p$ delineates the upper bound on the reconstruction error (measured in the $p$-norm) incurred by the reconstructor $\Psi$ when applied to noiseless data (when $\delta=0$). Conversely, $C^\epsilon_{\Psi, p}$ serves as a metric quantifying the amplification (or reduction), provided by $\Psi$, of input noise, which is upper bounded within a $p$-norm by $\epsilon$ in the data,  
since Definition \ref{def:stability_constant} can be rewritten as:
\begin{align}\label{eq:consequence_stab_constant}
    || \Psi(\K\x + \e) - \x ||_p \leq \eta_p + C^\epsilon_{\Psi, p} ||\e ||_p, \quad \forall \: || \e ||_p \leq \epsilon, \: \x \in \X.
\end{align} 

Now, we need to prove a couple of results.  
\begin{proposition}\label{prop:lip_of_w}
    For any $\delta > 0$ and any $\eta > 0$, $\w_\eta(\tilde{\x})$ defined in Equation \eqref{eq:OurWeights} is a Lipschitz continuous function of $| \D\tilde{\x} |$ in 1-norm, i.e. there exists $L(\w_\eta) \geq 0$, such that for any $\tilde{\x}, \tilde{\x}' \in \X$:
    \begin{align*}
        || \w_\eta(\tilde \x) - \w_\eta(\tilde{\x}') ||_1 \leq L(\w_\eta) || | \D \tilde{\x} | - | \D \tilde{\x}' | ||_1.
    \end{align*}
\end{proposition}

\begin{proof}
    If we consider:
    \begin{align*}
        f(\alpha) = \left( \frac{\eta}{\sqrt{\eta^2 + \alpha^2}} \right)^{1-p},
    \end{align*}
    the proof trivially follows from the observation that $f'(\alpha)$ is a continuous function of $\alpha$. Indeed:
    \begin{align*}
        f'(\alpha) =  (p-1) \frac{\alpha}{\sqrt{\eta^2 + \alpha^2}} \left( \frac{\eta}{\sqrt{\eta^2 + \alpha^2}} \right)^{p-2},
    \end{align*}
    which is continuous everywhere if $\eta \neq 0$.
\end{proof}

From now on, we denote as:
\begin{align}\label{eq:J_GT}
    \mathcal{J}_{GT}(\x, \y^\delta) := \frac{1}{2} ||\K\x-\y^\delta||_2^2 + \lambda || \w_\eta(\x^{GT}) \odot | \D \x | ||_1.
\end{align}
the objective function of our regularized model, where the ground truth image $\x^{GT}$ is used in place of $\tilde\x$ in the weighted strategy. 
We can thus prove the following inequality.
\begin{proposition}\label{prop:distance_of_J}
    For any $\x \in \X$, any fixed $\delta > 0$ and any $\y^\delta \in \R^m$, it holds:
    \begin{align}
        | \mathcal{J}_{GT}(\x, \y^\delta) -\mathcal{J}(\x, \y^\delta) | \leq \lambda L(\w_\eta) || \D ||_{2, 1} \left( \eta_1 + C^\epsilon_{\Psi, 1} ||\e||_1 \right) || \D \x ||_{2, 1},
    \end{align}
    where $\eta_1$ and $C^\epsilon_{\Psi, 1}$ are the accuracy and the $\epsilon$-stability constant in 1-norm of $\Psi$, respectively, $\epsilon>0$ is the maximum norm of the noise $\e$ in $\y^\delta$, and $\mathcal{J}_{GT}$ is given by Equation \eqref{eq:J_GT}.
\end{proposition}

\begin{proof}
    We first note that:
    \begin{align*}
        \mathcal{J}_{GT}(\x, \y^\delta) -\mathcal{J}(\x, \y^\delta)  
        =
        \lambda\left(|| \w_\eta(\x^{GT}) \odot | \D \x | ||_1 - || \w_\eta(\Psi(\y^\delta)) \odot | \D \x | ||_1\right).
    \end{align*}
    Thus,
    \begin{align*}
        | 	\mathcal{J}_{GT}(\x, \y^\delta) -\mathcal{J}(\x, \y^\delta) | & 
        = \lambda|\ || \w_\eta(\x^{GT}) \odot | \D \x | ||_1 - || \w_\eta(\Psi(\y^\delta)) \odot | \D \x | ||_1 \ | \\
        &
        \leq \lambda|| \w_\eta(\x^{GT}) \odot | \D \x | - \w_\eta(\Psi(\y^\delta)) \odot | \D \x | ||_1
    \end{align*}
    by the reverse triangular inequality.\\
    Now, let $\boldsymbol{W}_\eta(\tilde{\x})$ be the $n \times n$ diagonal matrix, whose diagonal entries are the elements of $\w_\eta(\tilde{\x})$, so that $\w_\eta(\tilde{\x}) \odot | \D \x | = \boldsymbol{W}_\eta(\tilde{\x}) | \D \x |$.
    Consequently, it holds:
    \begin{align*}
            | 	\mathcal{J}_{GT}(\x, \y^\delta) -\mathcal{J}(\x, \y^\delta) |
            &
            \leq \lambda \left\| \boldsymbol{W}_\eta(\x^{GT}) | \D \x | - \boldsymbol{W}_\eta(\Psi(\y^\delta))  | \D \x | \right\|_1 \\ 
            &
            = \lambda \left\| \left( \boldsymbol{W}_\eta(\x^{GT}) - \boldsymbol{W}_\eta(\Psi(\y^\delta)) \right) | \D \x |  \right\|_1 \\ 
            &
            \leq \lambda \left\| \boldsymbol{W}_\eta(\x^{GT}) - \boldsymbol{W}_\eta(\Psi(\y^\delta))  \right\|_1 \left\| | \D \x | \right\|_1.
    \end{align*}
    From Proposition \ref{prop:lip_of_w}, it follows:
    \begin{align*}
        \left\| \boldsymbol{W}_\eta(\x^{GT}) - \boldsymbol{W}_\eta(\Psi(\y^\delta))  \right\|_1 
        \leq 
        L(\w_\eta) \left\| | \D \x^{GT} | - | \D \Psi(\y^\delta) | \right\|_1,
    \end{align*}
    where $L(\w_\eta)$ is the Lipschitz constant of $\w_\eta$.
    
    To go on, note that the reverse triangular inequality gives:
    \begin{align*}
        \left\| | \D \x^{GT} | - | \D \Psi(\y^\delta) | \right\|_1 
        \leq 
        \left\| |  \D \x^{GT} - \D \Psi(\y^\delta)  | \right\|_1
        = 
        \left\| \D \left(\x^{GT} -\Psi (\y^\delta)\right) \right\|_{2,1},
    \end{align*}
    where $|| \cdot ||_{2, 1}$ is defined in \eqref{eq:TV_definition}. Now, let $|| \D ||_{2, 1}$ be the matrix norm induced by the $|| \cdot ||_{2, 1}$ vector norm, i.e.:
    \begin{align*}
        || \D ||_{2, 1} := \max_{\x \in \R^n} \frac{|| \D \x ||_{2, 1}}{|| \x ||_1}.
    \end{align*}
    We get:
    \begin{align*}
        || \D \left(\x^{GT} -\Psi (\y^\delta)\right) ||_{2,1} \leq || \D ||_{2, 1} || \x^{GT} - \Psi(\y^\delta) ||_1.
    \end{align*}
    To conclude, let $\epsilon > 0$ be the maximum norm of $\e$ in $\y^\delta = \K \x^{GT} + \e$. By Equation \eqref{eq:consequence_stab_constant}, following from Definition \ref{def:stability_constant}:
    \begin{align*}
        || \x^{GT} - \Psi(\y^\delta) ||_1 = || \x^{GT} - \Psi(\K \x^{GT} + \e) ||_1 \leq \eta_1 + C^\epsilon_{\Psi, 1} || \e ||_1,
    \end{align*}
    where $\eta_1$ and $C^\epsilon_{\Psi, 1}$ are the accuracy and the $\epsilon$-stability constant in 1-norm of $\Psi$, respectively.
    In conclusion, we have obtained the following inequality:
    \begin{align*}
        | 	\mathcal{J}_{GT}(\x, \y^\delta) -\mathcal{J}(\x, \y^\delta) | \leq \lambda L(\w_\eta) || \D ||_{2, 1} \left( \eta_1 + C^\epsilon_{\Psi, 1} ||\e||_1 \right) \left\| | \D \x | \right\|_1,
    \end{align*}
    which proves the result.
\end{proof}

To proceed, we first consider the noiseless case (i.e. $\delta = 0$), for which it holds the following Proposition.

\begin{proposition}\label{prop:stability_delta0}
    Let $\x^*_{GT, 0}$ be the unique minimum of $\mathcal{J}_{GT}(\x, \y^\delta)$, defined in \eqref{eq:J_GT}, with $\delta = 0$. Let $\{ \Psi_k \}_{k \in \mathbb{N}}$ be a sequence of reconstructor, each with accuracy $\eta_{1, k}^{-1}$ in 1-norm. Let $\x^*_{k, 0}$ be the unique minimum of $\mathcal{J}_k(\x, \y^0) := \mathcal{J}(\x, \y^\delta)$ with $\Psi = \Psi_k$ and $\delta = 0$. Then, if $\eta_{1, k} \to 0$ for $k \to \infty$:
    \begin{align*}
        || \x^*_{GT, 0} - \x^*_{k, 0} ||_1 \to 0.
    \end{align*}
\end{proposition}

\begin{proof}
    If $\delta = 0$, then for any $k \in \mathbb{N}$, Proposition \ref{prop:distance_of_J} implies:
    \begin{align*}
        | \mathcal{J}_{GT}(\x, \y^0) - \mathcal{J}_k(\x, \y^0) | \leq \eta_{1, k} L(\w_\eta) || \D ||_{2, 1} || \D \x ||_{2, 1} &\leq \eta_{1, k} L(\w_\eta) || \D ||_{2, 1}^2 || \x ||_1 \\ &\leq \eta_{1, k} n L(\w_\eta) || \D ||_{2, 1}^2 || \x ||_{\infty}, 
    \end{align*}
    which, considering that $\eta_{1, k} \to 0$ as $k \to \infty$, implies that $\mathcal{J}_k(\x, \y^0) \to \mathcal{J}_{GT}(\x, \y^0)$ uniformly on every compact subset of $\X$. \\
    
    Let $C \subseteq \X$ be the closure of any set that contains $\x^*_{k, 0}$ for all $k \in \mathbb{N}$. $C$ is limited, indeed let $\x^\dagger \in \R^n$ be a solution of $\K \x = \y^0$, then for any $k \in \mathbb{N}$:
    \begin{align*}
        \lambda || \w_\eta(\Psi_k(\y^0)) \odot | \D \x^*_{k, 0} | ||_1 \leq \mathcal{J}_k(\x^*_k, \y^0) \leq \mathcal{J}_k(\x^\dagger, \y^0) = \lambda || \w_\eta(\Psi_k(\y^0)) \odot | \D \x^\dagger | ||_1,
    \end{align*}
    where the second inequality comes from the definition of minimizer. By Proposition \ref{prop:w_is_a_scale_term} $\w_\eta(\Psi_k(\y^0)) \leq \boldsymbol{1}$, thus:
    \begin{align*}
        \lambda || \w_\eta(\Psi_k(\y^0)) \odot | \D \x^*_{k, 0} | ||_1 \leq \lambda || \w_\eta(\Psi_k(\y^0)) \odot | \D \x^\dagger | ||_1 \leq \lambda || \D\x^\dagger ||_{2, 1} = c < \infty.
    \end{align*}
    Consequently,
    \begin{align*}
        \lim_{k \to \infty} \lambda || \w_\eta(\Psi_k(\y^0)) \odot | \D \x^*_{k, 0} | ||_1 < \infty,
    \end{align*}
    which, since $\lambda || \w_\eta(\Psi_k(\y^0)) \odot | \D \x^*_{k, 0} | ||_1$ is coercive, implies that $\{ \x_k \}_{k \in \mathbb{N}}$ is limited, and so is $C$. \\
    
    Since $C$ is closed and limited, it is compact, and thus $\mathcal{J}_k(\x, \y^0) \to \mathcal{J}_{GT}(\x, \y^0)$ uniformly in $C$. Since $\x^*_{k, 0}$ is unique by Theorem \ref{theo:unicita}, by Proposition \ref{prop:uniformconvergence_implies_convergenceminimizers}, $|| \x^*_{k, 0} - \x^*_{GT, 0} ||_1 \to 0$ as $k \to \infty$.
\end{proof}

We can also prove a convergence result for $\x^*_{GT, \delta}$ and $\x^*_{\Psi, \delta}$.
\begin{proposition}\label{prop:convergence_GT}
    Let $\{ \delta_k \}_{k \in \mathbb{N}}$ be any sequence of positive noise levels such that $\delta_k \to 0$ as $k \to \infty$. For any $k \in \mathbb{N}$, let $\x^*_{GT, \delta_k}$ be the unique minimizer of $\mathcal{J}_{GT}(\x, \y^{\delta_k})$. Then:
    \begin{align*}
        \lim_{k \to \infty} || \x^*_{GT, \delta_k} - \x^*_{GT, 0} ||_1 \to 0.
    \end{align*}
\end{proposition}

\begin{proof}
    Note that $\mathcal{J}_{GT}(\x, \y^{\delta_k})$ is continuous with respect to $\y^{\delta_k}$, since the term $\y^{\delta_k}$ appears in $\mathcal{J}_{GT}(\x, \y^{\delta_k})$ only via a quadratic term. Consequently,
    \begin{align}
        \lim_{k \to \infty} \mathcal{J}_{GT}(\x, \y^{\delta_k}) = \mathcal{J}_{GT}(\x, \y^{0}), \quad \forall \x \in \X.
    \end{align}
    To prove the result, we need to show that $\mathcal{J}_{GT}(\x, \y^{\delta_k})$ converges uniformly to $\mathcal{J}_{GT}(\x, \y^{0})$ in $\x$ on any compact, which together with the uniqueness of the minimizer, implies that $\x^*_{GT, \delta_k}$ converges to $\x^*_{GT, 0}$ as $k \to \infty$ by Proposition \ref{prop:uniformconvergence_implies_convergenceminimizers}. To this aim, we observe that for any $\x \in \X$:
    \begin{align*}
        \left| \mathcal{J}_{GT}(\x, \y^{\delta_k}) - \mathcal{J}_{GT}(\x, \y^{0}) \right| &=\left| \frac{1}{2}|| \K \x - \y^{\delta_k} ||_2^2 -  \frac{1}{2}|| \K \x - \y^{0} ||_2^2 \right| \\ &\leq \frac{1}{2}||  \K \x - \y^{\delta_k} - \left( \K \x - \y^0 \right) ||_2^2 \\ & = \frac{1}{2}|| \y^{\delta_k} - \y^0 ||_2^2 \leq \frac{1}{2}\delta_k^2.
    \end{align*}
    Since $\delta_k$ does not depend on $\x$, and $\delta_k \to 0$ as $k \to \infty$, then this implies that $\mathcal{J}_{GT}(\x, \y^{\delta_k})$ converges to $\mathcal{J}_{GT}(\x, \y^{0})$ uniformly. This proves that $\x^*_{GT, \delta_k}$ converges to $\x^*_{GT, 0}$ and, in particular,
    \begin{align*}
        \lim_{k \to \infty}|| \x^*_{GT, \delta_k} - \x^*_{GT, 0} ||_1 = 0.
    \end{align*}
\end{proof}

The convergence of $\x^*_{\Psi, \delta}$ to $\x^*_{\Psi, 0}$ as $\delta \to 0$ for a general reconstructor $\Psi$ is slightly harder, since in this case, the regularizer also depends on $\delta$ due to $\tilde{\x} = \Psi(\y^\delta)$. 

\begin{proposition}\label{prop:convergence_Psi}
    For any reconstructor $\Psi: \R^m \to \R^n$, let $\{ \delta_k \}_{k \in \mathbb{N}}$ be any sequence of positive noise levels such that $\delta_k \to 0$ as $k \to \infty$. For any $k \in \mathbb{N}$, let $\x^*_{\Psi, \delta_k}$ be the unique minimizer of $\mathcal{J}(\x, \y^{\delta_k})$. Then:
    \begin{align*}
        \lim_{k \to \infty} || \x^*_{\Psi, \delta_k} - \x^*_{\Psi, 0} ||_1 \to 0.
    \end{align*}
\end{proposition}

\begin{proof}
    We proceed as in the proof of Proposition \ref{prop:convergence_GT}, denoting by $\mathcal{R}_k(\x)$ the value of the regularizer with weights computed by $\Psi(\y^{\delta_k})$ for the sake of simplicity, namely  i.e. $\mathcal{R}_k(\x) = || \w_\eta(\Psi(\y^{\delta_k})) \odot | \D \x | ||_1$. 
    \begin{align*}
        \left| \mathcal{J}(\x, \y^{\delta_k}) - \mathcal{J}(\x, \y^{0}) \right| &= \left| \frac{1}{2}|| \K \x - \y^{\delta_k} ||_2^2 -  \frac{1}{2}|| \K \x - \y^{0} ||_2^2 + \lambda \left( \mathcal{R}_k(\x) - \mathcal{R}_0(\x) \right) \right| \\ &\leq \left| \frac{1}{2}|| \K \x - \y^{\delta_k} ||_2^2 -  \frac{1}{2}|| \K \x - \y^{0} ||_2^2 \right| + \lambda \left| \mathcal{R}_k(\x) - \mathcal{R}_0(\x) \right| \\ &\leq \frac{1}{2}\delta_k^2 + \lambda \left| \mathcal{R}_k(\x) - \mathcal{R}_0(\x) \right|,
    \end{align*}
    The second term can be rearranged by means of Proposition \ref{prop:lip_of_w},
    \begin{align*}
        \left| \mathcal{R}_k(\x) - \mathcal{R}_0(\x) \right| &= \left| || \w_\eta(\Psi(\y^{\delta_k})) \odot | \D \x | ||_1 - ||\w_\eta(\Psi(\y^{0})) \odot | \D \x | ||_1 \right| \\ & \leq || \w_\eta(\Psi(\y^{\delta_k})) \odot | \D \x | - \w_\eta(\Psi(\y^{0})) \odot | \D \x | ||_1 \\ &= || \left(\w_\eta(\Psi(\y^{\delta_k})) - \w_\eta(\Psi(\y^{0}))\right) \odot | \D \x | ||_1 \\ &\leq || \w_\eta(\Psi(\y^{\delta_k})) - \w_\eta(\Psi(\y^{0})) ||_1 || \D \x ||_{2,1} \\ &\leq || \w_\eta(\Psi(\y^{\delta_k})) - \w_\eta(\Psi(\y^{0})) ||_1 || \D ||_{2,1} || \x ||_1 \\ &\leq || \w_\eta(\Psi(\y^{\delta_k})) - \w_\eta(\Psi(\y^{0})) ||_1 || \D ||_{2,1} n || \x ||_\infty \\ &\leq L(\w_\eta) || \Psi(\y^{\delta_k}) - \Psi(\y^0)||_1 || \D ||_{2,1} n || \x ||_\infty \\ &\leq L(\w_\eta) L(\Psi) \delta_k || \D ||_{2,1} n || \x ||_\infty,
    \end{align*}
    where $L(\Psi)$ is the Lipschitz constant of $\Psi$, which is well defined since $\Psi$ is Lipschitz continuous by definition.

    To summarize, we proved that:
    \begin{align*}
        \left| \mathcal{J}(\x, \y^{\delta_k}) - \mathcal{J}(\x, \y^{0}) \right| \leq \frac{1}{2}\delta_k^2 + \lambda L(\w_\eta) L(\Psi) \delta_k || \D ||_{2,1} n || \x ||_\infty,
    \end{align*}
    which implies that $\mathcal{J}(\x, \y^{\delta_k})$ converges uniformly to $\mathcal{J}(\x, \y^{0})$ on any compact subset of $\X$. By considering the closed subset $C$ of $\X$ that contains the sequence $\{ \x^*_{\Psi, \delta_k} \}_{k \in \mathbb{N}}$, we conclude as in Proposition \ref{prop:convergence_GT}, by observing that $C$ is limited since:
    \begin{align*}
        || \w_\eta(\Psi(\y^{\delta_k})) \odot | \D \x^*_{\Psi, \delta_k} | ||_1 &\leq \mathcal{J}_\Psi(\x^*_{\Psi, \delta_k}, \y^{\delta_k}) \leq \mathcal{J}_\Psi(\x^\dagger, \y^{\delta_k}) \\ & \leq \frac{\delta_k^2}{2} + \lambda || \w_\eta(\Psi(\y^{\delta_k})) \odot | \D \x^\dagger | ||_1 \\& \leq \frac{\delta_k^2}{2} + \lambda || \D \x^\dagger ||_{2, 1} \to c < \infty,
    \end{align*}
    where $\x^\dagger$ is any solution of $\K\x = \y^0$, and $|| \w_\eta(\Psi(\y^{\delta_k})) \odot | \D \x^*_{\Psi, \delta_k} | ||_1$ is coercive.
\end{proof}

In conclusion, we observe that the stability results from Proposition \ref{prop:stability_delta0} can be generalized to the case where $\delta > 0$. Thus, we can now prove Theorem \ref{teo:conv}.


\begin{proof}
    Let:
    \begin{align*}
        \mathcal{J}_k(\x, \y^\delta) := \frac{1}{2} || \K\x - \y^\delta ||_2^2 + \lambda || \w_\eta(\Psi_k(\y^\delta)) \odot | \D \x | ||_1.
    \end{align*}
    By proceeding as in the proof of Proposition \ref{prop:distance_of_J}, we can show that:
    \begin{align*}
        | \mathcal{J}_k(\x, \y^\delta) - \mathcal{J}_{GT}(\x, \y^\delta) | \leq \lambda L(\w_\eta) \left\| | \D \x^{GT} | - | \D \Psi_k(\y^\delta) | \right\|_1 || \D \x ||_{2, 1},
    \end{align*}
    where $L(\w_\eta)$ is defined as in Proposition \ref{prop:lip_of_w}. Consequently:
    \begin{align*}
        | \mathcal{J}_k(\x, \y^\delta) - \mathcal{J}_{GT}(\x, \y^\delta) | &\leq n \lambda L(\w_\eta) || \D ||_{2, 1} || \x ||_\infty \left\| | \D \x^{GT} | - | \D \Psi_k(\y^\delta) | \right\|_1,
    \end{align*}
    which implies that $\mathcal{J}_k(\x, \y^\delta)$ converges to $\mathcal{J}_{GT}(\x, \y^\delta)$ uniformly on any compact set, whenever
    \begin{align*}
        \lim_{k \to \infty} \left\| | \D \x^{GT} | - | \D \Psi_k(\y^\delta) | \right\|_1 = 0.
    \end{align*}
    By Proposition \ref{prop:uniformconvergence_implies_convergenceminimizers}, we conclude that:
    \begin{align*}
        \lim_{k \to \infty} || \x^*_{GT, \delta} - \x^*_{\Psi_k, \delta} ||_1 = 0.
    \end{align*}
\end{proof}

\end{appendices}

\bibliographystyle{plain}
\bibliography{biblio}
\end{document}